\documentclass[conference]{IEEEtran}

\usepackage[utf8]{inputenc}
\usepackage[english]{babel}
\usepackage{cite}
\usepackage{amsmath,amssymb,amsfonts}
\usepackage{subfig}    
\usepackage[demo]{graphicx}
\usepackage{textcomp}
\usepackage{xcolor}
\usepackage{mathrsfs}
\usepackage[utf8]{inputenc}
\usepackage{comment}
\usepackage[ruled,vlined]{algorithm2e}
\usepackage{mathtools}
\usepackage[font=small,labelfont=bf]{caption}
\usepackage{blkarray, bigstrut}
\usepackage{subfig}
\usepackage{cases}

\SetCommentSty{mycommfont}
\SetKwInput{KwInput}{Input}                % Set the Input
\SetKwInput{KwOutput}{Output}              % set the Output

\usepackage{enumitem}
\usepackage{float}

\usepackage{pgf, tikz}
\usetikzlibrary{arrows, automata}
\usepackage{tikz}
\usetikzlibrary{automata,positioning}

\newcommand*\circled[2][1.6]{\tikz[baseline=(char.base)]{
    \node[shape=circle, draw, inner sep=1pt, 
        minimum height={\f@size*#1},] (char) {\vphantom{WAH1g}#2};}}

\newlist{legal}{enumerate}{10}
\setlist[legal]{label*=\arabic*.}
\usepackage{ifpdf}

\usepackage{amsthm}

\newtheoremstyle{mystyle1}%                % Name
  {.5pt}%                                 % Space above
  {.5pt}%                                  % Space below
  {}%                                     % Body font
  {}%                                     % Indent amount
  {\bfseries}%                            % Theorem head font
  {.}%                                    % Punctuation after theorem head
  { }%                                    % Space after theorem head, ' ', or \newline
  {\thmname{#1}\thmnumber{ #2}\thmnote{ (#3)}}%                                     % Theorem head spec (can be left empty, meaning `normal')

\theoremstyle{mystyle1}

\newtheorem{defn}{Definition}

\newtheorem{rem}{Remark}

\newtheoremstyle{mystyle}%                % Name
  {.5pt}%                                 % Space above
  {.5pt}%                                 % Space below
  {\itshape}%                             % Body font
  {}%                                     % Indent amount
  {\bfseries}%                            % Theorem head font
  {.}%                                    % Punctuation after theorem head
  { }%                                    % Space after theorem head, ' ', or \newline
  {\thmname{#1}\thmnumber{ #2}\thmnote{ (#3)}}%                                     % Theorem head spec (can be left empty, meaning `normal')

 \theoremstyle{mystyle}

\newtheorem{thm}{Theorem}
\newtheorem{lem}{Lemma}
\newtheorem{prop}{Proposition}

\usepackage[utf8]{inputenc}
\usepackage[english]{babel}
\usepackage{cite}
\usepackage{amsmath,amssymb,amsfonts}
\usepackage{algorithmic}
\usepackage{graphicx}
\usepackage{textcomp}
\usepackage{xcolor}
\usepackage{mathrsfs} 
\usepackage{enumitem}
\usepackage{amsthm}
\usepackage{cite}
\usepackage{array}
\usepackage{amsmath}

\setlist[legal]{label*=\arabic*.}

\theoremstyle{mystyle1}

\usepackage[bottom=4.2cm]{geometry}
\usepackage{enumitem}
\usepackage{geometry}
\allowdisplaybreaks

\begin{document}

\title{Broadcast Rate Requires Nonlinear Coding in a Unicast Index Coding Instance of Size 36}

% author names and affiliations
\author{\IEEEauthorblockN{Arman Sharififar, Parastoo Sadeghi, Neda Aboutorab}
\IEEEauthorblockA{\textit{School of Engineering and Information Technology},
\textit{University of New South Wales, Australia} \\
Email:\{a.sharififar, p.sadeghi, n.aboutorab\}@unsw.edu.au
}
}

\maketitle

\begin{abstract}
Insufficiency of linear coding for the network coding problem was first proved by providing an instance which is solvable only by nonlinear network coding (Dougherty  \textit{et al.}, 2005). Based on the work of Effros,  \textit{et al.}, 2015, this specific network coding instance can be modeled as a groupcast index coding (GIC) instance with 74 messages and 80 users (where a message can be requested by multiple users). This proves the insufficiency of linear coding for the GIC problem. Using the systematic approach proposed by Maleki \textit{et al.}, 2014, the aforementioned GIC instance can be cast into a unicast index coding (UIC) instance with more than 200 users, each wanting a unique message. This confirms the necessity of nonlinear coding for the UIC problem, but only for achieving the entire capacity region. Nevertheless, the question of whether nonlinear coding is required to achieve the symmetric capacity (broadcast rate) of the UIC problem remained open. In this paper, we settle this question and prove the insufficiency of linear coding, by directly building a UIC instance with only 36 users for which there exists a nonlinear index code outperforming the optimal linear code in terms of the broadcast rate.
\end{abstract}

% \begin{abstract}
% Insufficiency of linear coding was first proved in the context of network coding by providing a network instance which is solvable only by nonlinear network coding (Dougherty  \textit{et al.}, 2005). Having established the connection between network coding and index coding \PS{(Rouayheb  \textit{et al.}, 2010) why not more recent and general work [3]?}, the network instance can be modeled as a groupcast index coding (GIC) instance with 82 messages and 175 users (where a message can be \PS{requested} by multiple users). This proves the insufficiency of linear coding for the GIC problem. A systematic approach was established (Maleki \textit{et al.}, 2014) to construct an equivalent unicast index coding (UIC) \PS{instance} for an arbitrary \PS{instance of the} GIC problem. This construction projects the \PS{aforementioned} GIC instance into a UIC instance with more than 200 users \PS{and confirms necessity of nonlinear coding in unicast setting, but only for achieving the entire capacity region}. The question of whether nonlinear coding is required to achieve the symmetric capacity region (broadcast rate) of the UIC problems remained open. In this paper, we settle this question and prove the insufficiency of linear code, by characterizing a UIC instance with only 36 users for which there exists a nonlinear index code outperforming the optimal linear coding in terms of the broadcast rate.
% \end{abstract}

\section{introduction}
%Index coding problem consists of a single server, having $m$ messages and a number of users, requesting one specific message from the server, and may aware of some other messages a priori, known as its side information. 

Index coding problem was first introduced by Birk and Kol \cite{Birk1998} in the context of satellite communication in which there is a single server, broadcasting $m$ messages to a number of users via a noiseless shared channel. Each user requests one specific message from the server and may already know some other messages as its side information. Exploiting the side information of the users, the server might be able to reduce the overall number of coded messages in order to communicate the $m$ messages \cite{Arman-3-1,Arman-3-2}. The main objective can be summarized as finding the minimum number of transmissions so that all users will be able to decode their requested message. The simple model established in index coding problem can be employed to study several important communication settings, including network coding \cite{Rouayheb2010,Effros2015}, distributed storage \cite{Li2018}, coded caching \cite{Maddah-ali2014, Wan2020}, and topological interference management \cite{Jafar2014, Maleki2014}.

The connection between network coding and index coding problem was established in \cite{Rouayheb2010}, in which a reduction method was provided to map any instance of network coding to a corresponding instance of index coding. This connection between network coding and index coding was extended in \cite{Effros2015} to include the general encoding and decoding functions so that the solution of an index coding instance will be suitably converted as a solution for the equivalent network coding instance.% However, this reduction technique turns any instance of network coding to to an instance of index coding such that some messages are requested by more than one user.

While in the unicast index coding (UIC), each message can be requested by only one user, this scenario can be extended to allow multiple users to request the same message, which is referred to as groupcast index coding (GIC). In fact, the aforementioned equivalence, turns any instance of network coding to an instance of index coding such that some messages are requested by more than one user.
%Index coding problem has been vigorously investigated in the literature, and various coding techniques for both UIC \cite{} and GIC \cite{} have been proposed to provide a solution for any specific index coding problem. However, the UIC setting has so far attracted more attention as its solutions and techniques can be generalized to include the GIC problem as well. 
In \cite{Maleki2014}, a systematic approach was proposed to construct an equivalent UIC instance for an arbitrarily GIC instance such that each message requested by multiple users is mapped to multiple distinct auxiliary messages requested by only one user. However, this construction method requires that the rate of auxiliary messages to be different from the rate of other messages, resulting in an asymmetric rate UIC problem.

Index coding schemes are broadly categorized into linear and nonlinear codes. Although linear index coding has been the center of attention due to their straightforward encoding and decoding processes, for the general index coding problem, they can be outperformed by nonlinear codes \cite{Dougherty2005,Rouayheb2010, Effros2015,Arman-4-2, Arman-4-3}. Insufficiency of linear codes was first proved in \cite{Dougherty2005} in the context of network coding by providing a network instance which is not solvable by any linear codes, while it can be solved by a nonlinear network code. Having established an efficient method of reducing a network coding instance to a corresponding instance of index coding \cite{Rouayheb2010,Effros2015}, the necessity of nonlinear coding was proved for the GIC problem. Furthermore, the construction technique of mapping an arbitrary GIC instance to an equivalent UIC instance \cite{Maleki2014} means that linear coding is also insufficient for the UIC problem, but only for achieving the entire capacity region.

\textbf{Open Problem \cite[Remark 1]{Maleki2014}}: For the UIC problems, the question of whether linear coding is necessary to achieve the broadcast rate remained as an open problem, despite the fact that broadcast rate is  most commonly used and studied in the index coding literature. 
\\
In this paper, we settle this open question by constructing a UIC instance for which there exists a nonlinear index code that can outperform the optimal linear coding in terms of the broadcast rate. This UIC instance consists of two separate UIC subinstances, which are connected to each other in some specific ways. In fact, the characterization of these subinstances is inspired by the two network coding subinstances in \cite{Dougherty2005}, where for the first one, linear coding cannot be optimal over any finite field with odd characteristic (i.e., a field with odd cardinality), and for the second one, linear coding is not able to achieve the broadcast rate over any finite field with characteristic two (i.e., field with even cardinality). In this work, the side information set of each user is designed, aiming at requiring the optimal index coding solution to satisfy the same specific constraints which must be met by any network code solving the network coding counterexample in \cite{Dougherty2005}.

One main merit of the UIC instance built in this paper, is its considerable simplicity with regard to the number of users (or messages) compared to the instances obtained based on mapping methods in \cite{Effros2015} and \cite{Maleki2014}. In fact, mapping the network coding counterexample in \cite{Dougherty2005} by the reduction method in \cite{Effros2015} results in a GIC instance, including 74 messages and 80 users (indeed, 148 users if each user requests only one specific message). This GIC instance using the mapping in \cite{Maleki2014} is turned into its equivalent asymmetric rate UIC instance, consisting of more than 200 users. However, the UIC instance provided in this paper comprises only 36 users.

\subsection{Summary of Contributions}
\begin{enumerate}[leftmargin=*]
    \item By directly designing a UIC instance, we prove that there exists a nonlinear index code which outperforms the optimal linear coding to achieve the symmetric capacity rate. This implies that linear coding is insufficient for the general UIC problem for achieving the broadcast rate. This UIC instance consists of two distinct subinstances which are connected in two different ways (Section \ref{sec:03}).
    \item We prove that for the first subinstance, its broadcast rate is not achievable by any linear coding over finite field with odd characteristic, while it can be achieved by a scalar binary linear code (Section \ref{sec:04}).
    \item For the second subinstance, we prove that linear coding over any finite field with characteristic two will not be able to achieve its broadcast rate. However, we show that there exists a scalar binary nonlinear index code which is optimal (Section \ref{sec:05}).
\end{enumerate}

\section{System Model and Background} \label{sec:02}
This section provides an overview of the system model and relevant background and definitions in index coding problem.
\subsection{Notation}
Scalar small letters such as $n$ denote an integer number where  $[n]:=\{1,...,n\}$. Scalar capital letters such as $L$ denote a set, whose cardinality is denoted by $|L|$ and $n_{L}:=\{n_l, l\in L\}$. Symbols in bold face such as $\boldsymbol{l}$ and $\boldsymbol{L}$ denote a vector and a matrix, respectively, with $\boldsymbol{L}^T$ denoting the transpose of matrix $\boldsymbol{L}$.  A calligraphic symbol such as $\mathcal{L}$ is used to denote a set whose elements are sets.\\
We write $\mathcal{X}$ to denote a finite alphabet. We use $\mathbb{F}_q$ to denote a finite field of size $q$ and write $\mathbb{F}_{q}^{n\times m}$ to denote the vector space of all $n\times m$ matrices over the field $\mathbb{F}_{q}$. Throughout the paper, $\boldsymbol{I}_n$ denotes the $n\times n$ identity matrix, and $\boldsymbol{0}_{n\times m}$ represents a full-zero matrix of size $n \times m$.

\subsection{System Model}
Consider a broadcast communication system in which a server transmits a set of $mt$ messages $X=\{x_{i,j},\ i\in[m],\ j\in [t]\},\ x_{i,j}\in \mathcal{X}$, to a number of users $U=\{u_i,\ i\in[m]\}$ through a noiseless broadcast channel. Each user $u_i$ wishes to receive a message of length $t$, $X_i=\{x_{i,j},\  j\in[t]\}$ and may have a prior knowledge of a subset of the messages $S_i:=\{x_{l,j},\ l\in A_{i},\ j\in[t]\},\ A_{i}\subseteq[m]\backslash \{i\}$, which is referred to as its side information set. The main objective is to minimize the number of coded messages which is required to be broadcast so as to enable each user to decode its requested message.
An instance of index coding problem $\mathcal{I}$ can be characterized by either the side information set of all users as $\mathcal{I}=\{A_i, i\in[m]\}$, or by their interfering message set $B_i=[m]\backslash (A_i \cup \{i\})$ as $\mathcal{I}=\{B_i, i\in[m]\}$.

\subsection{General Index Code}
\begin{defn}[$\mathcal{C}_{\mathcal{I}}$: Index Code for $\mathcal{I}$]
Given an instance of index coding problem $\mathcal{I}=\{A_i, i\in[m]\}$, a $(t,r)$ index code is defined as $\mathcal{C}_{\mathcal{I}}=(\phi_{\mathcal{I}},\{\psi_{\mathcal{I}}^{i}\})$, where
 \begin{itemize}[leftmargin=*]
     \item $\phi_{\mathcal{I}}: \mathcal{X}^{mt}\rightarrow \mathcal{X}^{r}$ is the encoder function which maps the $mt$ message symbol $x_{i,j}\in \mathcal{X}$ to the $r$ coded messages as $Y=\{y_1,\dots,y_r\}$, where $y_k\in \mathcal{X}, \forall k\in [r]$.
     \item $\psi_{\mathcal{I}}^{i}:$ represents the decoder function, where for each user $u_i, i\in[m]$, the decoder $\psi_{\mathcal{I}}^{i}: \mathcal{X}^{r}\times \mathcal{X}^{|A_i|t}\rightarrow \mathcal{X}^{t}$ maps the received $r$ coded messages $y_k\in Y, k\in[r]$ and the $|A_i|t$ messages $x_{l,j}\in S_i$ in the side information to the $t$ messages $\psi_{\mathcal{I}}^{i}(Y,S_i)=\{\hat{x}_{i,j}, j\in [t]\}$, where $\hat{x}_{i,j}$ is an estimate of $x_{i,j}$.
 \end{itemize}
\end{defn}

\begin{defn}[$\beta(\mathcal{C}_{\mathcal{I}})$: Broadcast Rate of $\mathcal{C}_{\mathcal{I}}$]
Given an instance of index coding problem $\mathcal{I}$, the broadcast rate of a $(t,r)$ index code $\mathcal{C}_{\mathcal{I}}$ is defined as $\beta(\mathcal{C}_{\mathcal{I}})=\frac{r}{t}$.
\end{defn}

\begin{defn}[$\beta_{\mathcal{I}}$: Broadcast Rate of $\mathcal{I}$]
Given an instance of index coding problem $\mathcal{I}$, the broadcast rate $\beta_{\mathcal{I}}$ is defined as
\begin{equation} \label{eq:opt-rate}
    \beta_{\mathcal{I}}=\inf_{t} \inf_{\mathcal{C}_{\mathcal{I}}} \beta(\mathcal{C}_{\mathcal{I}}).
\end{equation}
\end{defn}

\begin{defn}[Scalar and Vector Index Code]
The index code $\mathcal{C}$ is considered to be scalar if $t=1$. Otherwise, it is called a vector code. For scalar codes, we use $x_i=x_{i,1}, \forall i\in [m]$, for simplicity.
\end{defn}
\subsection{Linear Index Code}
We can assume that the finite alphabet $\mathcal{X}$ is selected as a finite field $\mathbb{F}_q$, so linear operations are well-defined. Let $\boldsymbol{x}=[\boldsymbol{x}_1,\dots,\boldsymbol{x}_m]^{T}\in \mathbb{F}_{q}^{mt\times 1}$ denote the message vector, where $\boldsymbol{x}_{i}=[x_{i,1},\dots,x_{i,t}]\in \mathbb{F}_{q}^{1\times t}$ is the requested message vector by user $u_i, \forall i\in[m]$. We denote the side information vector of $u_i$ as $\boldsymbol{x}_{A_i}=[\boldsymbol{x}_{1}^{i},\dots, \boldsymbol{x}_{m}^{i}]^{T}$, where $\boldsymbol{x}_{j}^{i}=\boldsymbol{x}_{j}$, if $j\in A_i$ and $\boldsymbol{x}_{j}^{i}=\boldsymbol{0}_{1\times t}$, otherwise. 

\begin{defn}[$\mathcal{L}_{\mathcal{I}}$: Linear Index Code for $\mathcal{I}$]
Given an instance of index coding problem $\mathcal{I}=\{A_i, i\in[m]\}$, a $(t,r)$ linear index code is defined as $\mathcal{L}_{\mathcal{I}}=(\boldsymbol{H},\{\psi_{\mathcal{L}}^{i}\})$, where
  \begin{itemize}[leftmargin=*]
      \item $\boldsymbol{H}: \mathbb{F}_{q}^{mt\times 1}\rightarrow \mathbb{F}_{q}^{r\times 1}$ is the encoder matrix which maps the message vector $\boldsymbol{x}=[\boldsymbol{x}_1,\dots,\boldsymbol{x}_m]^{T}\in \mathbb{F}_{q}^{mt\times 1}$  to a coded message vector $\boldsymbol{y}=[y_1,\dots,y_r]^{T}\in \mathbb{F}_{q}^{{r}\times 1}$ as follows
      \begin{equation} \nonumber
      \boldsymbol{y}=\boldsymbol{H}\boldsymbol{x}=\sum_{i\in [m]} \boldsymbol{H}_{i}\boldsymbol{x}_i^{T}.
      \end{equation}
      Here $\boldsymbol{H}_{i}\in \mathbb{F}_{q}^{r\times t}$ is the encoder matrix of the $i$-th message vector $\boldsymbol{x}_i$ such that
      $\boldsymbol{H}=
      \left [\begin{array}{c|c|c}
        \boldsymbol{H}_1 & \dots & \boldsymbol{H}_m
      \end{array}
     \right ]\in \mathbb{F}_{q}^{r\times mt}$.
     \item $\psi_{\mathcal{I}}^{i}$ represents the linear decoder function for user $u_{i}, i\in[m]$, where $\psi_{\mathcal{I}}^{i}(\boldsymbol{y}, \boldsymbol{x}_{A_i})$ maps the received coded message $\boldsymbol{y}$ and its side information vector $\boldsymbol{x}_{A_i}$ to $\hat{\boldsymbol{x}}_{i}$, which is an estimate of the requested message vector $\boldsymbol{x}_i$.
  \end{itemize}
\end{defn}

\begin{prop}
It can be shown \cite{Sharififar2021} that the necessary and sufficient condition for linear decoder $\psi_{\mathcal{I}}^{i}, \forall i\in[m]$ to be able to correctly decode the requested message vector $\boldsymbol{x}_i$ is
   \begin{equation} \label{eq:dec-cond}
       \mathrm{rank} \ \boldsymbol{H}_{\{i\}\cup B_i}= \mathrm{rank} \ \boldsymbol{H}_{B_i} + t,
   \end{equation}
where $\boldsymbol{H}_L$ denotes the matrix $\left [\begin{array}{c|c|c}
       \boldsymbol{H}_{l_1} & \dots & \boldsymbol{H}_{l_{|L|}}
     \end{array}
    \right ]$ for the given set $L=\{l_1,\dots,l_{|L|}\}$. %such that $l_1<\dots<l_{|L|}$
\end{prop}

\begin{defn}[$\lambda_{q}(\mathcal{L}_{\mathcal{I}})$: Linear Broadcast Rate of $\mathcal{L}_{\mathcal{I}}$ over $\mathbb{F}_q$]
Given an instance of index coding problem $\mathcal{I}$, the linear broadcast rate of a $(t,r)$ index code $\mathcal{L}_{\mathcal{I}}$ over field $\mathbb{F}_q$ is defined as $\lambda_{q}(\mathcal{L}_{\mathcal{I}})=\frac{r}{t}$.
\end{defn}

\begin{defn}[$\lambda_{\mathcal{I},q}$: Linear Broadcast Rate of $\mathcal{I}$ over $\mathbb{F}_q$]
Given an instance of index coding problem $\mathcal{I}$, the linear broadcast rate $\lambda_{\mathcal{I},q}$ over field $\mathbb{F}_q$ is defined as
    \begin{equation} \nonumber
        \lambda_{\mathcal{I},q}=\inf_{t} \inf_{\mathcal{L}_{\mathcal{I}}} \lambda_{q}(\mathcal{L}_{\mathcal{I}}).
    \end{equation}
\end{defn}

\begin{defn}[$\boldsymbol{H}_{\mathcal{I},q}^{\ast}$: Optimal Encoder Matrix of $\mathcal{I}$ over $\mathbb{F}_q$]
Given an instance of index coding problem $\mathcal{I}$, an encoder matrix, satisfying the decoding condition in \eqref{eq:dec-cond} for all users, is said to be optimal over $\mathbb{F}_q$ if the number of its rows is equal to ${\lambda_{\mathcal{I},q}}t$. Such an optimal encoder matrix is denoted by $\boldsymbol{H}_{\mathcal{I},q}^{\ast}$.%, and can be obtained by choosing any $\lambda_{\mathcal{I},q}$ linearly independent rows of $\boldsymbol{H}$ in \eqref{eq:encoder-minrank} such that its submatrices $\boldsymbol{H}_{ij}$ gives the solution of $\eqref{eq:minrank-opt}$ while meeting the constraints in $\eqref{eq:minrank-constraints}$.
\end{defn}

\begin{defn}[$\lambda_{\mathcal{I}}$: Linear Broadcast Rate for $\mathcal{I}$]
Given an instance of index coding problem $\mathcal{I}$, the linear broadcast rate is defined as 
\begin{equation} \label{eq:opt-lin-rate}
    \lambda_{\mathcal{I}}=\min_{q} \lambda_{\mathcal{I},q}.
\end{equation}

%which can be achieved as follows
%   \begin{equation} \label{eq:opt-lin-rate}
%       \lambda_{\mathcal{I}}= \min_{q} \ \mathrm{minrank}_{\mathbb{F}_q}(\mathcal{I})
%   \end{equation}
\end{defn}

\begin{rem}
As it can be observed from \eqref{eq:opt-rate}, for the broadcast rate $\beta_{\mathcal{I}}$, the minimization is not taken over the alphabet size $|\mathcal{X}|$, because it is proved that the broadcast rate is independent of the chosen alphabet size $|\mathcal{X}|$ \cite{Cannons2005,Arbabjolfaei2018}. In contrast, the linear broadcast rate $\lambda_{\mathcal{I}}$ in \eqref{eq:opt-lin-rate} is achieved by searching through all the possible field sizes $q$, recognizing its dependency on $q$. In fact, given a field size $q_1$, reference \cite{Lubetzky2009} provides a method to construct an index coding instance $\mathcal{I}$ to prove that there exists another field size $q_2$ such that $\lambda_{\mathcal{I},q_2}<\lambda_{\mathcal{I},q_1}$, and this gap can be significant for sufficiently large number of users. Such a dependency on the field size is the underlying reason causing the linear codes to be insufficient for achieving the broadcast rate in general for the index coding problem.
\end{rem}

\subsection{Graph Definitions}
\begin{defn}[$\mathcal{G}_{\mathcal{I}}$: Graph Representation of $\mathcal{I}$]
The index coding instance $\mathcal{I}=\{A_i, i\in[m]\}$ can be represented as a directed graph $\mathcal{G}_{\mathcal{I}}=(V,E)$, where $V=[m]$ and $E\subseteq[m]\times [m]$, respectively, denote the vertex and edge set such that $(i,j)\in E$ if and only if (iff) $i\in A_j$, for all $(i,j)\in [m]\times [m]$. Graph $\mathcal{G}_{\mathcal{I}}^{\prime}=(V^{\prime}, E^{\prime})$ is considered as an induced subgraph of $\mathcal{G}_{\mathcal{I}}$, where $V^{\prime}\subseteq[m]$ and $E^{\prime}\subseteq V^{\prime}\times V^{\prime}$ such that $(i,j)\in E^{\prime}$ iff $i\in A_j\cap V^{\prime}$. 
%$V^{\prime}$ is said to be cyclic or aIf $\mathcal{G}_{\mathcal{I}}^{\prime}$ forms an acyclic subgraph
\end{defn}

\begin{defn}[Acyclic and Independent Set of $\mathcal{I}$]
If the subgraph $\mathcal{G}_{\mathcal{I}}^{\prime}=(V^{\prime}, E^{\prime})$ forms an acyclic graph, then $V^{\prime}$ is said to be acyclic set of $\mathcal{I}$. If thic acyclic subgraph has no edges, i.e., $E^{\prime}=\emptyset$, then $V^{\prime}$ is called the independent set of $\mathcal{I}$.
\end{defn}

\begin{defn}[Minimal Cyclic Set of $\mathcal{I}$]
If subgraph $\mathcal{G}_{\mathcal{I}}^{\prime}=(V^{\prime}, E^{\prime})$ forms a cycle such that any of its induced subgraphs $\mathcal{G}_{\mathcal{I}}^{\prime\prime}=(V^{\prime\prime}, E^{\prime\prime})$, $V^{\prime\prime}\subset V^{\prime}$, is an acyclic graph, then $V^{\prime}$ is called the minimal cyclic set of $\mathcal{I}$.
\end{defn}

\begin{defn}[Maximum Acyclic Induced Subgraph (MAIS) of $\mathcal{I}$]
If $\mathcal{G}_{\mathcal{I}}^{\prime}=(V^{\prime}, E^{\prime})$ forms an acyclic subgraph with the maximum size of $V^{\prime}$, then $V^{\prime}$ is considered as a MAIS set of $\mathcal{I}$, and $\beta_{\mathrm{MAIS}(\mathcal{I})}=|V^{\prime}|$ is called the MAIS bound on $\mathcal{I}$.
\end{defn}

\begin{prop}[\mdseries Bar-Yossef \textit{et all}. \cite{Bar-Yossef2011}]
Given the index coding instance $\mathcal{I}$ and a finite field $\mathbb{F}_q$, we have
\end{prop}
 
\begin{equation}
    \lambda_{\mathcal{I},q}\geq \lambda_{\mathcal{I}}\geq \beta_{\mathcal{I}}\geq \beta_{\mathrm{MAIS}(\mathcal{I})}.
\end{equation}

%Note that, any finite field has either even cardinality (i.e., characteristic two) or odd cardinality (i.e., odd characteristic).

\begin{comment}
\begin{prop}[\mdseries Blasiak \textit{et all}. \cite{Blasiak2011}]
 Let $\beta_{\mathcal{I}_1}$ and $\beta_{\mathcal{I}_2}$ denote the optimal broadcast rate of $\mathcal{I}_1$ and $\mathcal{I}_2$, respectively. Then, for the optimal broadcast rate of $\mathcal{I}_3=\mathcal{I}_1 \nleftrightarrow \mathcal{I}_2$ and $\mathcal{I}_4=\mathcal{I}_1 \leftrightarrow \mathcal{I}_2$, we have
 \begin{equation} \nonumber
    \left\{
      \begin{array}{cc}
         \beta_{\mathcal{I}_3}&=\beta_{\mathcal{I}_1}+\beta_{\mathcal{I}_2}, \ \ \ \ \ \ \ \ \ \\ \\
        \beta_{\mathcal{I}_4}&=\max \{\beta_{\mathcal{I}_1}+\beta_{\mathcal{I}_2}\}.
     \end{array}
     \right.
\end{equation}
\end{prop}
\end{comment}

\section{Insufficiency of Linear Coding for the UIC Problem in Terms of Broadcast Rate} \label{sec:03}
This section gives a description of the two specific ways of connecting the two instances $\mathcal{I}_1$ and $\mathcal{I}_2$, which will be used in this paper for proving the insufficiency of linear coding for achieving the broadcast rate of the UIC problem.

\begin{defn}[$\mathcal{I}_1 \nleftrightarrow \mathcal{I}_2$: No-way Connection of $\mathcal{I}_1$ and $\mathcal{I}_2$]
Given two index coding instances $\mathcal{I}_1=\{A_{i}^{1}, i\in [m_1]\}$ and $\mathcal{I}_2=\{A_{i}^{2}, i\in [m_2]\}$, no-way connection of $\mathcal{I}_1$ and $\mathcal{I}_2$, denoted by $\mathcal{I}_1 \nleftrightarrow \mathcal{I}_2$, is defined as a new index coding instance $ \mathcal{I}=\{A_i, i\in[m]\}$, where $m=m_1+m_2$ and
\begin{equation} \nonumber
    \left\{
      \begin{array}{cc}
        A_i&=A_{i}^{1}, \ \ \ \forall i\in [m_1], \\ \\
        A_{i+m_1}&=A_{i}^{2}, \ \ \ \forall i\in [m_2]. 
     \end{array}
     \right.
\end{equation}
\end{defn}

\begin{defn}[$\mathcal{I}_1 \leftrightarrow \mathcal{I}_2$: Two-way Connection of $\mathcal{I}_1$ and $\mathcal{I}_2$]
Given two index coding instances $\mathcal{I}_1=\{A_{i}^{1}, i\in [m_1]\}$ and $\mathcal{I}_2=\{A_{i}^{2}, i\in [m_2]\}$, two-way connection of $\mathcal{I}_1$ and $\mathcal{I}_2$, denoted by $\mathcal{I}_1 \leftrightarrow \mathcal{I}_2$, is defined as a new index coding instance $ \mathcal{I}=\{A_i, i\in[m]\}$, where $m=m_1+m_2$ and
\begin{equation} \nonumber
    \left\{
      \begin{array}{cc}
        A_i&=A_{i}^{1} \cup ([m]\backslash[m_1]),\ \ \ \ \ \ \ \ \ \forall i\in [m_1], \\ \\
        A_{i+m_1}&=A_{i}^{2} \cup [m_1], \ \ \ \ \ \quad \quad \quad \quad \forall i\in [m_2],
     \end{array}
     \right.
\end{equation}
which means that the new instance $\mathcal{I}$ is a concatenation of the two subinstances $\mathcal{I}_1$ and $\mathcal{I}_2$ such that each user in $\mathcal{I}_1$ has all the messages requested by the users in $\mathcal{I}_2$ in its side information set and vice versa.
\end{defn}

\begin{prop}[\mdseries Blasiak \textit{et all}. \cite{Blasiak2011}] \label{prop:no-two}
Let $\lambda_{\mathcal{I}_1,q}$ and $\lambda_{\mathcal{I}_2,q}$, respectively, denote the linear broadcast rate of $\mathcal{I}_1$ and $\mathcal{I}_2$ over $\mathbb{F}_q$. Then, for the linear broadcast rate of $\mathcal{I}_3=\mathcal{I}_1 \nleftrightarrow \mathcal{I}_2$ and $\mathcal{I}_4=\mathcal{I}_1 \leftrightarrow \mathcal{I}_2$ over $\mathbb{F}_q$, we have
 \begin{equation} \nonumber
    \left\{
      \begin{array}{cc}
         \lambda_{\mathcal{I}_3,q}&=\lambda_{\mathcal{I}_1,q}+\lambda_{\mathcal{I}_2,q}, \ \ \ \ \ \ \ \ \ \\ \\
        \lambda_{\mathcal{I}_4,q}&=\max \{\lambda_{\mathcal{I}_1,q}, \lambda_{\mathcal{I}_2,q}\}.
     \end{array}
     \right.
\end{equation}
\end{prop}

\begin{thm} \label{thm:main}
Linear coding is insufficient for achieving the broadcast rate of the UIC problems.
\end{thm}
\begin{proof}
In the next two sections, two UIC instances $\mathcal{I}_1$ and $\mathcal{I}_2$, respectively, with $m_1=10$ and $m_2=26$ messages will be characterized with the following properties:
\begin{itemize}[leftmargin=*]
    \item In Theorem \ref{thm:01}, first, we show that $\lambda_{\mathcal{I}_1,2}=\beta_{\mathcal{I}_1}=6$ by designing a scalar binary linear code $\mathcal{L}_{\mathcal{I}_1}=(\boldsymbol{H}_{\mathcal{I}_1,2}^{\ast}, \{\psi_{\mathcal{I}_1}^{i}\})$. This implies the optimality of the binary linear code. Then, we prove that $\lambda_{\mathcal{I}_1,q}>6, \forall q=2k+1, k\geq 1$, which means that linear coding over any finite field with odd characteristic is not optimal.
    \item In Theorem \ref{thm:02}, first, we prove that $\lambda_{\mathcal{I}_2,q}>\beta_{\mathcal{I}_2}=6, \forall q=2k, k\geq 1$, which means that linear coding over any finite field with characteristic two is not optimal. Then, we show that there exists a scalar binary nonlinear code $\mathcal{C}_{\mathcal{I}_2}=(\phi_{\mathcal{I}_2}, \{\psi_{\mathcal{I}_2}^{i}\})$, which is optimal, i.e., $\beta(\mathcal{C}_{\mathcal{I}_2})=6$.
\end{itemize}
This implies that, if the finite field $\mathbb{F}_q$ has either characteristic two or odd characteristic, then one of $\lambda_{\mathcal{I}_1,q}$ and $\lambda_{\mathcal{I}_2,q}$ will always be greater than 6.
Hence, according to \eqref{eq:opt-lin-rate} and Proposition \ref{prop:no-two}, we have
\begin{equation}
     \left\{
      \begin{array}{cc}
        \lambda_{\mathcal{I}_3}&=\min_{q}  (\lambda_{\mathcal{I}_1,q} + \lambda_{\mathcal{I}_2,q})>12, \ \ \ \ \ \ \ \\ \\
        \lambda_{\mathcal{I}_4}&=\min_{q}  \max \{\lambda_{\mathcal{I}_1,q}, \lambda_{\mathcal{I}_2,q}\}>6.
     \end{array}
     \right.
\end{equation}
Now, let $\boldsymbol{y}=[y_1,\dots,y_{6}]^T= \boldsymbol{H}_{\mathcal{I}_1,2}^{\ast}[x_1,\dots,x_{10}]^T$ and $\{z_1,\dots,z_{6}\}=\phi_{\mathcal{I}_2}(x_{[36]\backslash[10]})$. Then, we have $\beta(\mathcal{C}_{\mathcal{I}_3})=12$ and $\beta(\mathcal{C}_{\mathcal{I}_4})=6$ for two scalar nonlinear codes $\mathcal{C}_{\mathcal{I}_3}$ and $\mathcal{C}_{\mathcal{I}_4}$ such that
\begin{equation}
     \left\{
      \begin{array}{cc}
        \phi_{\mathcal{I}_3}(x_{[36]})=\{y_1,\dots,y_6\}\cup\{z_1,\dots,z_6\},
        \\
        \\
        \phi_{\mathcal{I}_4}(x_{[36]})=\{y_1\oplus z_1,\dots, y_6\oplus z_6\}, \ \ \ \ \ \
     \end{array}
     \right.
\end{equation}
which completes the proof.
\end{proof}

\subsection{Prerequisite Material for Proof of Theorems \ref{thm:01} and \ref{thm:02}}

\begin{rem} \label{rem:dec-con}
It can be easily verified that the decoding condition in \eqref{eq:dec-cond} along with the properties of the $rank$ function give the following results.
\begin{align}
    \mathrm{rank}\ \boldsymbol{H}_{\{i\}\cup B_{i}^{\prime}}&= \mathrm{rank}\ \boldsymbol{H}_{B_{i}^{\prime}} + t, \ \ \ \forall B_{i}^{\prime}\subseteq B_{i}, \ \forall i\in [m],  \label{eq:rem:dec-con-1} 
    \\
    \mathrm{rank}\ \boldsymbol{H}_{i}&= t, \ \ \ \ \ \ \ \ \ \ \ \ \ \ \ \ \ \ \ \forall i\in [m], \label{eq:rem:dec-con-2}
    \\
    \mathrm{rank}\ \boldsymbol{H}_{L_1}&\leq \mathrm{rank}\ \boldsymbol{H}_{L_2}, \ \ \ \ \ \ \ \ \forall L_1 \subseteq L_2 \subseteq [m]. \label{eq:rem:dec-con-3}
\end{align}
\end{rem}

\begin{comment}
\begin{rem}
Let $B_{l}^{L}=B_{l}\cap L$ be the $L$-local interfering message set of user $u_l, l\in L\subseteq [m]$. Now, if $L$ is
\begin{itemize}[leftmargin=*]
    \item a minimal cyclic set, then $L\backslash \{l\}$ will become an acyclic set for all $l\in L$.
    \item an independent set, then $B_{l}^{L}=L\backslash \{l\}$ for all $l\in L$.
    \item an acyclic set, then, we can always find a sequence of its elements $l_1, \dots, l_{|L|}\in L$ such that $L_i\subseteq B_{l_i}, \forall l_i\in L$, where $L_i=L\backslash \{l_1,\dots,l_i\}$.
\end{itemize}
\end{rem}
\end{comment}

\begin{lem} \label{lem:MAIS1}
Assume $L\subseteq [m]$ is an acyclic set of $\mathcal{I}$. Then, the condition in \eqref{eq:rem:dec-con-1} for all $i\in L$ requires $\mathrm{rank} \ \boldsymbol{H}_L=|L|t$.
\end{lem}
%\PS{Proof is provided in \cite{} and is omitted here due to space limitation. Remove for archive. }
\begin{proof}
If $L$ is an acyclic set, then we can find a sequence of its elements $l_1,\dots,l_{|L|}\in L$ such that $L_j\subseteq B_{l_j}, \forall j\in [|L|]$, where $L_j=\{l_{j+1},\dots,l_{|L|}\}, \forall j\in [|L|-1]$ and $L_{|L|}=\emptyset$. Note $L=\{l_1\}\cup L_1$ and $L_j=\{l_{j+1}\}\cup L_{j+1}, \forall j\in [|L|-1]$. By applying the condition in \eqref{eq:rem:dec-con-1} for each $i=l_1,\dots,l_{|L|}$, we have
\begin{align}
    \mathrm{rank}\ \boldsymbol{H}_{L=\{l_1\}\cup L_1}&= \mathrm{rank}\ \boldsymbol{H}_{L_1=\{l_2\}\cup L_2}+t \nonumber\\
    &= \mathrm{rank}\ \boldsymbol{H}_{L_2=\{l_3\}\cup L_3}+2t \nonumber \\
    &=\dots \nonumber \\
    &= |L|t \nonumber.
\end{align}
\vspace{-2ex}
\end{proof}

\begin{lem} \label{lem:min-cyc}
Let $L\subseteq[m]$ be a minimal cyclic set of $\mathcal{I}$ with encoder matrix $\boldsymbol{H}$. In order to have $\mathrm{rank}\ \boldsymbol{H}_L=(|L|-1)t$, we must have $\boldsymbol{H}_{l}=\sum_{j\in L\backslash\{l\}} \boldsymbol{H}_{j} \boldsymbol{M}_{l,j},\forall l\in L$ such that each $\boldsymbol{M}_{l,j}\in \boldsymbol{F}_{q}^{t\times t}$ is invertible.
\end{lem}

\begin{proof}
First, note that for any $l\in L$, set $L\backslash\{l\}$ is an acyclic set. Then, according to Lemma \ref{lem:MAIS1}, 
\begin{equation} \label{eq:lem3}
    \mathrm{rank}\ \boldsymbol{H}_{L\backslash\{l\}}=(|L|-1)t, \ \ \forall l\in L.
\end{equation}
So, having $\mathrm{rank}\ \boldsymbol{H}_L=(|L|-1)t$ requires $\boldsymbol{H}_{l}=\sum_{j\in L\backslash\{l\}} \boldsymbol{H}_{j} \boldsymbol{M}_{l,j}$. Now, if one of the $\boldsymbol{M}_{l,j}, j\in L\backslash \{l\}$ is not invertible, then $\mathrm{rank}\ \boldsymbol{H}_{L\backslash\{j\}}<(|L|-1)t$, which contradicts \eqref{eq:lem3}. Thus, all $\boldsymbol{M}_{l,j}$ must be invertible.
\end{proof}
%\begin{defn}[Independent  Set]
%$L$ is said to be an independent set, if for each $j\in L$ we have $L\backslash \{i\}\subseteq B_i, \forall i\in L$. Note, any independent set is an acyclic set as well.
%\end{defn}

\begin{lem} \label{lem:MAIS2}
Assume $L\subseteq[m]$ is an independent set of $\mathcal{I}$ and let $ j\in [m]\backslash L$. Now, if $j\in B_i,\forall i\in L\backslash \{l\}$ for some $l\in L$, then in order to have $\mathrm{rank}\ \boldsymbol{H}_{\{j\}\cup L}=|L|t$, one must satisfy $\boldsymbol{H}_j=\boldsymbol{H}_l\boldsymbol{M}_{j,l}$ for some invertible matrix $\boldsymbol{M}_{j,l}$.
\end{lem}

\begin{proof}
First, because $L$ is an independent set, then $L\backslash \{i\}\subseteq B_i, \forall i\in L$. Moreover since $L$ is an acyclic set, then $\mathrm{rank}\ \boldsymbol{H}_{L}=|L|t$. So, in order to have $\mathrm{rank}\ \boldsymbol{H}_{\{j\}\cup L}=|L|t$, we must have $\boldsymbol{H}_j= \sum_{l\in L} \boldsymbol{H}_l\boldsymbol{M}_{j,l}$. Let $ j\in [m]\backslash L$ and $j\in B_i$ for some $i \in  L\backslash \{l\}$ which leads to $\{j\}\cup L\backslash \{i\}\subseteq B_i$. Now, assume $\boldsymbol{M}_{j,i}$ is a nonzero matrix (so, $\mathrm{rank}\ \boldsymbol{M}_{j,i}\geq 1$). Then, 
\begin{align}
    \mathrm{rank}\ \boldsymbol{H}_{\{j\}\cup L}&=
    \mathrm{rank}\ \boldsymbol{H}_{\{i\}\cup (\{j\}\cup  L\backslash\{i\})}
    \nonumber
    \\
    &=\mathrm{rank}\ \boldsymbol{H}_{\{j\}\cup L\backslash \{i\}}+t \label{lem:pr:01} 
    \\
    &=\mathrm{rank}\ \left [\begin{array}{c|c}
        \boldsymbol{H}_j & \boldsymbol{H}_{L\backslash \{i\}}
      \end{array}
     \right ]+t
     \nonumber
    \\
    &=\mathrm{rank}\ \left [\begin{array}{c|c}
        \sum_{l\in L} \boldsymbol{H}_l\boldsymbol{M}_{j,l} & \boldsymbol{H}_{L\backslash \{i\}}
      \end{array}
     \right ]+t
     \nonumber
    \\
    &\geq \mathrm{rank}\ \left [\begin{array}{c|c}
        \boldsymbol{H}_i\boldsymbol{M}_{j,i} & \boldsymbol{H}_{L\backslash \{i\}}
      \end{array}
     \right ]+t 
     \label{lem:pr:02} 
    \\
    &=\mathrm{rank}\ \boldsymbol{H}_{i}\boldsymbol{M}_{j,i}+ (|L|-1)t+t
    \label{lem:pr:03}
    \\
    &> |L|t, 
    \label{lem:pr:04}
\end{align}
where \eqref{lem:pr:01} is due to \eqref{eq:rem:dec-con-1}, \eqref{lem:pr:02} is because of the property of the $\mathrm{rank}$ function by removing the term $\sum_{l\in L\backslash \{i\}} \boldsymbol{H}_l\boldsymbol{M}_{j,l}$ from $\sum_{l\in L} \boldsymbol{H}_l\boldsymbol{M}_{j,l}$ as it is a linear combination of the columns of $\boldsymbol{H}_{L\backslash \{i\}}$. \eqref{lem:pr:03} is based on Lemma \ref{lem:MAIS1} and the fact that $L$ is an acyclic set. Thus, the column space of $\boldsymbol{H}_i$ is linearly independent of column space of $\boldsymbol{H}_{L\backslash \{i\}}$. Finally, \eqref{lem:pr:04} is due to the fact that $\boldsymbol{H}_i$ is full-rank and $\mathrm{rank}\ \boldsymbol{M}_{j,i}\geq 1$. The result in \eqref{lem:pr:04} contradicts the assumption that $\mathrm{rank}\ \boldsymbol{H}_{\{j\}\cup L}=|L|t$, and hence, we must have $\boldsymbol{M}_{j,i}=0$. The same argument for $i\in L\backslash \{l\}$ gives $\boldsymbol{M}_{j,i}=0, \forall i\in L\backslash \{l\}$. So, $\boldsymbol{H}_j=\boldsymbol{H}_l \boldsymbol{M}_{j,l}$ and $\boldsymbol{M}_{j,l}$ must be invertible to have $\mathrm{rank}\ \boldsymbol{H}_j=t$.
\end{proof}

\begin{comment}
Now, assume that $\lambda_{\mathcal{I}_1,q}=\lambda_{\mathcal{I}_2,q}$. Then, $\lambda_{\mathcal{I}_3,q}$ is achieved by choosing $\boldsymbol{H}_{2}^{\prime}=\boldsymbol{M}\boldsymbol{H}_{\mathcal{I}_1,q}^{\ast}$ and $\boldsymbol{H}_{1}^{\prime}=\boldsymbol{M}^{-1}\boldsymbol{H}_{\mathcal{I}_2,q}^{\ast}$ for some invertible matrix $\boldsymbol{M}$. Therefore,
\begin{equation}
    \boldsymbol{H}_{\mathcal{I}_4,q}^{\ast}= \boldsymbol{H}_{\mathcal{I}_1,q}^{\ast} + \boldsymbol{M}^{-1}\boldsymbol{H}_{\mathcal{I}_2,q}^{\ast}.
\end{equation}
\end{comment}

\section{The Index Coding Instance $\mathcal{I}_1$} \label{sec:04}
In this section, we provide the index coding instance $\mathcal{I}_1$, where its broadcast rate is $\beta_{\mathcal{I}_1}=6$. First, we provide a scalar binary linear code which is optimal. Then, we prove that this rate is not achievable by linear coding over any finite field with odd characteristic.

The index coding instance $\mathcal{I}_1=\{A_i, i\in[10]\}$ is characterized as follows
\begin{equation}
  \left \{
    \begin{array}{ccc}
     A_1&=\{4,6,7\}, \ A_2=\{1,5,6\}, \ A_3=\{1,2,7\},  
     \\
     A_4&=\{2,3,6\}, \ A_5=\{1,3,4\}, \ A_6=\{3,5,7\},  
     \\
     A_7&=\{2,4,5\}, \ A_8=A_9=A_{10}=\emptyset. \ \ \ \ \ \ \ \ \ \ \ \ \
    \end{array}
  \right.
\end{equation}

\begin{thm} \label{thm:01}
  For the index coding instance $\mathcal{I}_1$, $\lambda_{\mathcal{I}_1,2}=\beta_{\mathcal{I}_1}=6$. However,  $\lambda_{\mathcal{I}_1,q}=\beta_{\mathcal{I}_1}>6, \forall q=2k+1, k\geq 1$.
\end{thm}
Each claim is proved separately in Propositions \ref{prop:I1:binary} and \ref{prop:I1:odd}, respectively.
\begin{prop} \label{prop:I1:binary}
   $\lambda_{\mathcal{I}_1,2}=6$. This means Binary linear coding achieves the broadcast rate of $\mathcal{I}_1$.
\end{prop}

\begin{proof}
First, it can be observed that $V^{\prime}=\{1,2,3,8,9,10\}$ is a MAIS set of $\mathcal{I}_1$. So $\beta_{\mathrm{MAIS}(\mathcal{I}_1)}=6$. Now, it can be verified that the following scalar binary linear code achieves the MAIS bound: 
\begin{equation} \label{eq:GH}
    \left \{\begin{array}{cc}
        y_1&= x_1 \oplus x_4 \oplus x_6 \oplus x_7, \ \ \ \ \   y_4= x_8,\ \ \\
        y_2&= x_2 \oplus x_4 \oplus x_5 \oplus x_7,  \ \ \ \ \ y_5= x_9,\ \ \\
        y_3&= x_3 \oplus x_5 \oplus x_6 \oplus x_7,  \ \ \ \ \ y_6= x_{10}.
        %y_4&= x_8, \ \ \ \ \ \ \ \ \ \ \ \ \ \ \ \ \ \ \ \\
        %y_5&= x_9, \ \ \ \ \ \ \ \ \ \ \ \ \ \ \ \ \ \ \ \\
        %y_6&= x_{10}, \ \ \ \ \ \ \ \ \ \ \ \ \ \ \ \ \ \\
    \end{array} \right.
\end{equation}
Therefore, the binary linear code is optimal.
\end{proof}

\begin{prop} \label{prop:I1:odd}
$\lambda_{\mathcal{I}_1,q}>6$ for all $q= 2k+1, k\geq 1$. This means linear coding over any field with odd characteristic cannot achieve the broadcast rate of $\mathcal{I}_1$.
\end{prop}
%First we provide Lemma $\ref{lem:MAIS1}$ and \ref{lem:min-cyc}.

\begin{proof}
First, note that $\mathcal{I}_1=\mathcal{I}_1^{\prime} \nleftrightarrow \mathcal{I}_1^{\prime\prime}$, where $\mathcal{I}_1^{\prime}=\{A_i, i\in [7]\}$ and $\mathcal{I}_1^{\prime\prime}=\{A_8, A_9, A_{10}\}$. For the subinstance $\mathcal{I}_2^{\prime\prime}$, $L_1=\{8,9,10\}$ forms a MAIS set, so $\beta_{\mathrm{MAIS}(\mathcal{I}_2^{\prime\prime})}=3$, which means uncoded transmission $y_8=x_8$, $y_9=x_9$ and $y_{10}=x_{10}$ is optimal. Now, we prove that for the subinstance $\mathcal{I}_1^{\prime}=\{B_{i}^{\prime}, i\in [7]\}$, where $B_{i}^{\prime}=B_{i}\cap [7]$ (local interfering message set), any linear coding over a field with odd characteristic cannot be optimal.
It can be seen that set $L_2=\{1,2,3\}$ is a MAIS set of $\mathcal{I}_1^{\prime}$, so $\beta_{\mathrm{MAIS}(\mathcal{I}_{1}^{\prime})}=3$. To achieve $\lambda_{\mathcal{I}_{1}^{\prime},q}=3$,
the decoding condition in \eqref{eq:dec-cond} and Lemma \ref{lem:MAIS1}, respectively, give
\begin{align}
       \mathrm{rank} \ \boldsymbol{H}_{B_i^{\prime}}&=2t,  \ \ \ \ \ \ \ \forall i\in [7], 
       \label{eq:pr:prop:B_i}
       \\ 
         \mathrm{rank} \ ( \boldsymbol{H}_{\{1,2,3\}}&= 
    \left [\begin{array}{c|c|c}
             \boldsymbol{H}_{1} & \boldsymbol{H}_2 & \boldsymbol{H}_{3}
          \end{array}
    \right ])=3t.
    \label{eq:pr:prop:H_{123}}
\end{align}
So, to have $\mathrm{rank} \ \boldsymbol{H}=3t$, other $\boldsymbol{H}_i, i\in \{4,5,6,7\}$ must be expressed as a linear combination of $\boldsymbol{H}_1, \boldsymbol{H}_2,\boldsymbol{H}_3$.
\\
It can also be observed that each set $B_i^{\prime}, i\in [7]$ %$C_1=\{1,2,4\}$, $C_2=\{2,3,5\}$, $C_3=\{1,3,6\}$, $C_4=\{1,5,7\}$, $C_5=\{2,6,7\}$, $C_6=\{3,4,7\}$, and $C_7=\{4,5,6\}$ all
is a minimal cyclic set. Thus, because we want $\mathrm{rank} \ \boldsymbol{H}_{B_i^{\prime}}=2t$, then based on Lemma \ref{lem:min-cyc}, we will have seven constraints as follows (for simplicity, in this proof, each matrix $\boldsymbol{M}_{j}$ is numerated by only one index)
\begin{align}
     B_1^{\prime}&\rightarrow  \boldsymbol{H}_5= \boldsymbol{H}_2 \boldsymbol{M}_{1} +  \boldsymbol{H}_3 \boldsymbol{M}_{2},  \label{eq:pr:thm:5-1-2}\\ %\nonumber\\
     B_2^{\prime}&\rightarrow  \boldsymbol{H}_7= \boldsymbol{H}_3 \boldsymbol{M}_{3} +  \boldsymbol{H}_4 \boldsymbol{M}_{4},  \label{eq:pr:thm:7-3-4}\\ %\nonumber\\
     B_3^{\prime}&\rightarrow  \boldsymbol{H}_6= \boldsymbol{H}_4 \boldsymbol{M}_{5} +  \boldsymbol{H}_5 \boldsymbol{M}_{6}, \label{eq:pr:thm:6-4-5}\\ 
     B_4^{\prime}&\rightarrow  \boldsymbol{H}_7= \boldsymbol{H}_1 \boldsymbol{M}_{7} +  \boldsymbol{H}_5 \boldsymbol{M}_{8},  \label{eq:pr:thm:7-1-5}\\ %\nonumber\\
     B_5^{\prime}&\rightarrow  \boldsymbol{H}_7= \boldsymbol{H}_2 \boldsymbol{M}_{9} +  \boldsymbol{H}_6 \boldsymbol{M}_{10},  \label{eq:pr:thm:7-2-6}\\ %\nonumber\\
     B_6^{\prime}&\rightarrow  \boldsymbol{H}_4= \boldsymbol{H}_1 \boldsymbol{M}_{11} +  \boldsymbol{H}_2 \boldsymbol{M}_{12}, \label{eq:pr:thm:4-1-2} %\nonumber
     \\ %\nonumber\\
     B_7^{\prime}&\rightarrow  \boldsymbol{H}_6= \boldsymbol{H}_1 \boldsymbol{M}_{13} +  \boldsymbol{H}_3 \boldsymbol{M}_{14}, \label{eq:pr:thm:6-1-3}
\end{align}
where the matrices $\boldsymbol{M}_j\in \mathbb{F}_{q}^{t \times t}$, $j\in[14]$ must be all invertible. We show that meeting these seven constraints will lead to a contradiction over any field with odd characteristic. Now, in \eqref{eq:pr:thm:7-3-4}, \eqref{eq:pr:thm:7-1-5} and \eqref{eq:pr:thm:7-2-6}, we replace $\boldsymbol{H}_4, \boldsymbol{H}_5, \boldsymbol{H}_6$ with their equal term, respectively, in \eqref{eq:pr:thm:4-1-2}, \eqref{eq:pr:thm:5-1-2} and \eqref{eq:pr:thm:6-1-3}, which leads to
\begin{align}
    \eqref{eq:pr:thm:7-3-4},\ \eqref{eq:pr:thm:4-1-2}&\rightarrow \boldsymbol{H}_7= \boldsymbol{H}_1 \boldsymbol{M}_{11}\boldsymbol{M}_{4} +  \boldsymbol{H}_2 \boldsymbol{M}_{12}\boldsymbol{M}_{4} +  \boldsymbol{H}_3 \boldsymbol{M}_{3}, \label{eq:proof1:01}
    \\
    \eqref{eq:pr:thm:7-1-5}, \ \eqref{eq:pr:thm:5-1-2}&\rightarrow \boldsymbol{H}_7= \boldsymbol{H}_1 \boldsymbol{M}_{7} +  \boldsymbol{H}_2 \boldsymbol{M}_{1}\boldsymbol{M}_{8} +  \boldsymbol{H}_3 \boldsymbol{M}_{2}\boldsymbol{M}_{8}, \label{eq:proof1:02}
    \\ 
    \eqref{eq:pr:thm:7-2-6}, \ \eqref{eq:pr:thm:6-1-3}&\rightarrow \boldsymbol{H}_7= \boldsymbol{H}_1 \boldsymbol{M}_{13}\boldsymbol{M}_{10} +  \boldsymbol{H}_2 \boldsymbol{M}_{9} +  \boldsymbol{H}_3 \boldsymbol{M}_{14}\boldsymbol{M}_{10}. \label{eq:proof1:03}
\end{align}
Due to \eqref{eq:pr:prop:H_{123}}, it can be seen that \eqref{eq:proof1:01}, \eqref{eq:proof1:02} and \eqref{eq:proof1:03} are equal iff their coefficients of $\boldsymbol{H}_1, \boldsymbol{H}_3, \boldsymbol{H}_2$ will be equal. So, equating coefficients of $\boldsymbol{H}_1, \boldsymbol{H}_3, \boldsymbol{H}_2$, respectively, gives
\begin{align}
    \boldsymbol{M}_{11}\boldsymbol{M}_{4}&=\boldsymbol{M}_{13}\boldsymbol{M}_{10} \rightarrow \boldsymbol{M}_{4}= \boldsymbol{M}_{11}^{-1}\boldsymbol{M}_{13}\boldsymbol{M}_{10}, \label{eq:proof2:01}\\ %\nonumber\\
    \boldsymbol{M}_{2}\boldsymbol{M}_{8}&=\boldsymbol{M}_{14}\boldsymbol{M}_{10}\rightarrow \boldsymbol{M}_{8}= \boldsymbol{M}_{2}^{-1}\boldsymbol{M}_{14}\boldsymbol{M}_{10}, \label{eq:proof2:02}\\ %\nonumber
   \boldsymbol{M}_{12}\boldsymbol{M}_{4}&=\boldsymbol{M}_{1}\boldsymbol{M}_{8}. \label{eq:proof2:03}
\end{align}
Now, in \eqref{eq:proof2:03}, we substitute $\boldsymbol{M}_{4}$ and $\boldsymbol{M}_{8}$ with their equal term, respectively, in \eqref{eq:proof2:01} and \eqref{eq:proof2:02} which results in
\begin{equation}
\boldsymbol{M}_{12}\boldsymbol{M}_{11}^{-1}\boldsymbol{M}_{13}=\boldsymbol{M}_{1}\boldsymbol{M}_{2}^{-1}\boldsymbol{M}_{14}. \label{eq:proof3:01}
\end{equation}
On the other hand, in \eqref{eq:pr:thm:6-4-5}, we replace $\boldsymbol{H}_4, \boldsymbol{H}_5, \boldsymbol{H}_6$ with their equal term, respectively, in \eqref{eq:pr:thm:4-1-2}, \eqref{eq:pr:thm:5-1-2} and \eqref{eq:pr:thm:6-1-3}. Then, equating coefficients of $\boldsymbol{H}_1, \boldsymbol{H}_2, \boldsymbol{H}_3$ gives
\begin{align}
    \boldsymbol{M}_{13}&=\boldsymbol{M}_{11} \boldsymbol{M}_{5}, \label{eq:proof4:01} \\ %\nonumber
    %\\
    \boldsymbol{M}_{14}&=\boldsymbol{M}_2 \boldsymbol{M}_{6}, \label{eq:proof4:02} \\ %\nonumber
    %\\
    \boldsymbol{M}_{12}\boldsymbol{M}_{5}&+\boldsymbol{M}_{1} \boldsymbol{M}_{6}=0. \label{eq:proof4:03}
\end{align}
Now, in \eqref{eq:proof4:03}, we substitute $\boldsymbol{M}_{5}$ and $\boldsymbol{M}_{6}$ with their equal term in \eqref{eq:proof4:01}, \eqref{eq:proof4:02}, respectively, which gives
\begin{equation}
    \boldsymbol{M}_{12}\boldsymbol{M}_{11}^{-1}\boldsymbol{M}_{13}+\boldsymbol{M}_{1}\boldsymbol{M}_{2}^{-1}\boldsymbol{M}_{14}=0. \label{eq:proof5:01}
\end{equation}
Finally, since all the $\boldsymbol{M}_i$'s are invertible, from \eqref{eq:proof3:01} and \eqref{eq:proof5:01}, we must have $\boldsymbol{I}_t=-\boldsymbol{I}_t$, which is not possible over any field with odd characteristic. This completes the proof.
\end{proof}
%\PS{Good the proof is more insightful.}

\section{The Index coding instance $\mathcal{I}_2$} \label{sec:05}
This section provides the index coding instance $\mathcal{I}_2$ where its broadcast rate is $\beta_{\mathcal{I}_2}=6$. We prove that this rate is not achievable by linear coding over any finite field with characteristic two. However, we show that, there exists a nonlinear code over the binary field which can achieve the broadcast rate.

The index coding instance $\mathcal{I}_2=\{B_i, i\in[26]\}$ is characterized as follows.
\begin{equation}
  \left \{
    \begin{array}{ccc}
     B_1&=\{2,3,5,11,12,13,15,22,23,24,25,26\}, \\
     B_2&=\{1,3,6,11,12,13,16,21,22,24,25,26\}, \\
     B_3&=\{1,2,4,11,12,13,14,21,22,23,24,26\},  \\
     B_4&=\{5,6,14,15,16\}, \ \ \ \ \ \ \ \ \ \ \ \ \ \ \ \ \ \ \ \ \ \ \ \ \ \ \ \ \\
     B_5&=\{4,6,14,15,16\}, \ \ \ \ \ \ \ \ \ \ \ \ \ \ \ \ \ \ \ \ \ \ \ \ \ \ \ \ \\
     B_6&=\{4,5,14,15,16\}, \ \ \ \ \ \ \ \ \ \ \ \ \ \ \ \ \ \ \ \ \ \ \ \ \ \ \ \ \\
     B_7&=\{17\}, \ \ \ \ \ \ \ \ \ \ \ \ \ \ \ \ \ \ \ \ \ \ \ \ \ \ \ \ \ \ \ \ \ \ \ \ \ \ \ \ \  \\
     B_8&=\{1,5,7,11,15,17,18\},\ \  \ \ \ \ \ \ \ \ \ \ \ \ \ \ \ \ \ \ \ 
     \\
     B_{9}&=\{2,6,7,12,16,17,19\}, \ \ \ \ \ \ \ \ \ \ \ \ \ \ \ \ \ \ \ \ \ \\
     B_{10}&=\{3,4,7,13,14,17,20\}, \ \ \ \ \ \ \ \ \ \ \ \ \ \ \ \ \ \ \ \ \\
     B_{11}&=\{1,2,3,5,12,13,15,21,23,24,25,26\}, \\
     B_{12}&=\{1,2,3,6,11,13,16,21,22,23,25,26\}, \\
     B_{13}&=\{1,2,3,4,11,12,14,21,22,23,24,25\},  \\
     B_{14}&=\{4\}, \ \ \ \ \ \ \ \ \ \ \ \ \ \ \ \ \ \ \ \ \ \ \ \ \ \ \ \ \ \ \ \ \ \ \ \ \ \ \ \ \ \ \\
     B_{15}&=\{5\}, \ \ \ \ \ \ \ \ \ \ \ \ \ \ \ \ \ \ \ \ \ \ \ \ \ \ \ \ \ \ \ \ \ \ \ \ \ \ \ \ \ \ \\
     B_{16}&=\{6\}, \ \ \ \ \ \ \ \ \ \ \ \ \ \ \ \ \ \ \ \ \ \ \ \ \ \ \ \ \ \ \ \ \ \ \ \ \ \ \ \ \ \ \\
     B_{17}&=\{7\}, \ \ \ \ \ \ \ \ \ \ \ \ \ \ \ \ \ \ \ \ \ \ \ \ \ \ \ \ \ \ \ \ \ \ \ \ \ \ \ \ \ \ \\
     B_{18}&=\{1,5,7,8,11,15,17\},\ \  \ \ \ \ \ \ \ \ \ \ \ \ \ \ \ \ \ \ \ 
     \\
     B_{19}&=\{2,6,7,9,12,16,17\}, \ \ \ \ \ \ \ \ \ \ \ \ \ \ \ \ \ \ \ \ \ \\
     B_{20}&=\{3,4,7,10,13,14,17\}, \ \ \ \ \ \ \ \ \ \ \ \ \ \ \ \ \ \ \ \ \\
     B_{21}&=\{4,6,14,16,22\},\ \ \ \ \ \ \ \ \ \ \ \ \ \ \ \ \ \ \ \ \ \ \ \ \ \ \ \\
     B_{22}&=\{4,6,14,16,21\},\ \ \ \ \ \ \ \ \ \ \ \ \ \ \ \ \ \ \ \ \ \ \ \ \ \ \ \\
     B_{23}&=\{4,5,14,15,24\},\ \ \ \ \ \ \ \ \ \ \ \ \ \ \ \ \ \ \ \ \ \ \ \ \ \ \ \\
     B_{24}&=\{4,5,14,15,23\},\ \ \ \ \ \ \ \ \ \ \ \ \ \ \ \ \ \ \ \ \ \ \ \ \ \ \ \\
     B_{25}&=\{5,6,15,16,26\},\ \ \ \ \ \ \ \ \ \ \ \ \ \ \ \ \ \ \ \ \ \ \ \ \ \ \ \\
     B_{26}&=\{5,6,15,16,25\}.\ \ \ \ \ \ \ \ \ \ \ \ \ \ \ \ \ \ \ \ \ \ \ \ \ \ \ 
    \end{array}
  \right.
\end{equation}

\begin{thm} \label{thm:02}
$\lambda_{\mathcal{I}_2,q}>6$, $\forall q=2k, k\geq 1$. This means that linear coding over any finite field with characteristic two cannot achieve the broadcast rate of $\mathcal{I}_2$.
\end{thm}

\begin{proof}
First, it can be observed that $L=\{1,2,3,11,12,13\}$ is an independent set of $\mathcal{I}_2$. Thus, based on Lemma \ref{lem:MAIS1}, 
\begin{equation} \label{eq:pr:thm2:L-ind}
    \mathrm{rank}\ \boldsymbol{H}_{L}=6t.
\end{equation}
Now, if $\lambda_{\mathcal{I}_2,q}=6$, then $\mathrm{rank} \boldsymbol{H}=6t$, and thus, each $\boldsymbol{H}_j, j\in[26]\backslash L$ must be expressed as a linear combination of $\boldsymbol{H}_j, j\in L$ as $\boldsymbol{H}_j=\sum_{l\in L} \boldsymbol{H}_{l}\boldsymbol{M}_{j,l}$.
In the following, through \textbf{\textit{Steps 1 to 5}}, we determine the constraints which must be held on the linear space of each $\boldsymbol{H}_i, i\in [27]\backslash L$, which will finally lead to a contradiction if the field has characteristic two.
\\
\textbf{\textit{Step 1}}: In this step, we determine the constraints on the linear space of $\boldsymbol{H}_i, i\in\{21,...,26\}$.
\\
Lemma \ref{lem:MAIS2} gives the following results.
\begin{align} 
        21\in B_i, \forall i\in L\backslash \{1\} &\rightarrow \boldsymbol{H}_{21}=\boldsymbol{H}_{1}\boldsymbol{M}_{21,1}, \label{eq:pr:thm:21}
         \\ 
        22\in B_i, \forall i\in L\backslash \{11\} &\rightarrow \boldsymbol{H}_{22}=\boldsymbol{H}_{11}\boldsymbol{M}_{22,11}, \label{eq:pr:thm:22}
         \\ 
        23\in B_i, \forall i\in L\backslash \{2\} &\rightarrow \boldsymbol{H}_{23}=\boldsymbol{H}_{2}\boldsymbol{M}_{23,2}, \label{eq:pr:thm:23}
         \\
        24\in B_i, \forall i\in L\backslash \{12\} &\rightarrow \boldsymbol{H}_{24}=\boldsymbol{H}_{12}\boldsymbol{M}_{24,12}, \label{eq:pr:thm:24}
         \\
        25\in B_i, \forall i\in L\backslash \{3\} &\rightarrow \boldsymbol{H}_{25}=\boldsymbol{H}_{3}\boldsymbol{M}_{25,3}, \label{eq:pr:thm:25}
         \\ 
        26\in B_i, \forall i\in L\backslash \{13\} &\rightarrow \boldsymbol{H}_{26}=\boldsymbol{H}_{13}\boldsymbol{M}_{26,13}, \label{eq:pr:thm:26}
\end{align}
for some invertible matrices $\boldsymbol{M}_{21,1}$, $\boldsymbol{M}_{22,11}$, $\boldsymbol{M}_{23,2}$, $\boldsymbol{M}_{24,12}$, $\boldsymbol{M}_{25,3}$ and $\boldsymbol{M}_{26,13}$, respectively.
\\
\textbf{\textit{Step 2}}: In this step, we determine the constraints on the linear space of $\boldsymbol{H}_{\{i,i+10\}}, i\in\{4,5,6\}$, which we write them as
\begin{align}
    \boldsymbol{H}_{\{i,i+10\}}=
    \left [\begin{array}{c|c|c} \boldsymbol{H}_{\{1,11\}}\boldsymbol{N}_{i,1} & \boldsymbol{H}_{\{2,12\}}\boldsymbol{N}_{i,2} & \boldsymbol{H}_{\{3,13\}}\boldsymbol{N}_{i,3} \end{array} \right ],
\end{align}
where, 
\begin{equation}
    \boldsymbol{N}_{i,j}\triangleq
    \begin{bmatrix}
       \boldsymbol{M}_{i,j} & \boldsymbol{M}_{i+10,j}\\
       \boldsymbol{M}_{i,j+10} & \boldsymbol{M}_{i+10,j+10}
    \end{bmatrix}.
\end{equation} 
Let $C_{1}^{\prime}=\{2,3,12,13\}\subset C_1=\{2,3,5,12,13,15\}$, $C_{2}^{\prime}=\{1,3,11,13\}\subset C_2=\{1,3,6,11,13,16\}$ and $C_{3}^{\prime}=\{1,2,11,12\}\subset C_3=\{1,2,4,11,12,14\}$. Then, for all $i=1,2,3$, we have
\begin{align}
    6t&= \mathrm{rank} \ \boldsymbol{H}
    \label{eq:pr:thm2:C_i:1-main}
    \\
    &\geq \mathrm{rank} \ \boldsymbol{H}_{\{i\}\cup (\{i+10\}\cup C_i)}
    \label{eq:pr:thm2:C_i:1}
    \\
    &= \mathrm{rank} \ \boldsymbol{H}_{\{i+10\}\cup C_i} + t
    \label{eq:pr:thm2:C_i:2}
    \\ 
    &= \mathrm{rank} \ \boldsymbol{H}_{C_i} + 2t
    \label{eq:pr:thm2:C_i:3}
    \\
    &\geq \mathrm{rank} \ \boldsymbol{H}_{C_{i}^{\prime}} + 2t
    \label{eq:pr:thm2:C_i:4}
    \\
    &=6t,
    \label{eq:pr:thm2:C_i:5}
\end{align}
where \eqref{eq:pr:thm2:C_i:1-main} is because we desire $\lambda_{\mathcal{I}_2,q}=6$, \eqref{eq:pr:thm2:C_i:1} is due to \eqref{eq:rem:dec-con-3}, \eqref{eq:pr:thm2:C_i:2} and \eqref{eq:pr:thm2:C_i:3} are, respectively, because of \eqref{eq:rem:dec-con-1} and the fact that $\{i+10\}\cup C_i\subset B_i$ and $C_i\subset B_{i+10}$. \eqref{eq:pr:thm2:C_i:4} is due to $C_{i}^{\prime}\subset C_i$, and finally, \eqref{eq:pr:thm2:C_i:5} follows from the fact that each $C_{i}^{\prime}$ is an independent set. Now, based on \eqref{eq:pr:thm2:C_i:1-main}, \dots, \eqref{eq:pr:thm2:C_i:5}, we have
\begin{equation}
    \mathrm{rank} \ \boldsymbol{H}_{C_i}=\mathrm{rank} \ \boldsymbol{H}_{C_{i}^{\prime}}=4t, \ \ \ \ \forall i=1,2,3,
\end{equation}
which implies that each $\boldsymbol{H}_{C_{i}\backslash C_{i}^{\prime}}$ must be expressed as a linear combination of $\boldsymbol{H}_{C_{i}^{\prime}}$ for $i=1,2,3$. This, respectively, results in
\begin{align}
     \boldsymbol{H}_{\{5,15\}}&=\left [\begin{array}{c|c} \boldsymbol{H}_{\{2,12\}}\boldsymbol{N}_{5,2} & \boldsymbol{H}_{\{3,13\}}\boldsymbol{N}_{5,3} \end{array} \right],
     \label{eq:pr:thm2:5-15}
    \\ 
     \boldsymbol{H}_{\{6,16\}}&=\left [\begin{array}{c|c} \boldsymbol{H}_{\{1,11\}}\boldsymbol{N}_{6,1} & \boldsymbol{H}_{\{3,13\}}\boldsymbol{N}_{6,3} \end{array} \right],
    \label{eq:pr:thm2:6-16}
    \\
    \boldsymbol{H}_{\{4,14\}}&=\left [\begin{array}{c|c} \boldsymbol{H}_{\{1,11\}}\boldsymbol{N}_{4,1} & \boldsymbol{H}_{\{2,12\}}\boldsymbol{N}_{4,2} \end{array} \right].
    \label{eq:pr:thm2:4-14}
\end{align}
\textbf{\textit{Step 3}}: In this step, we show that each $\boldsymbol{N}_{5,2}$, $\boldsymbol{N}_{5,3}$, $\boldsymbol{N}_{6,1}$, $\boldsymbol{N}_{6,3}$, $\boldsymbol{N}_{4,1}$ and $\boldsymbol{N}_{4,2}$ is invertible.
\\
First, let $D_1=\{4,14,21,22\}$, $D_2=\{4,14,23,24\}$, $D_3=\{5,15,23,24\}$, $D_4=\{5,15,25,26\}$, $D_5=\{6,16,21,22\}$ and $D_6=\{6,16,25,26\}$. Then, for $D_1$, we have
    \begin{align}
    4t&=\mathrm{rank} \ \boldsymbol{H}_{D_1}
    \label{eq:top:-1}
    \\
    &=\mathrm{rank}\ \left [\begin{array}{c|c} \boldsymbol{H}_{\{4,14\}} & \boldsymbol{H}_{\{21,22\}} \end{array} \right]
    \nonumber%\label{eq:top:00}
    \\
     &=\mathrm{rank}\ \left [\begin{array}{c|c} \boldsymbol{H}_{\{4,14\}} & \boldsymbol{H}_{\{1,11\}} \begin{bmatrix} \boldsymbol{M}_{21,1} & \boldsymbol{M}_{22,11} \end{bmatrix}^T \end{array}  \right ]
     \label{eq:top:01}
    \\
    &=\mathrm{rank}\ \left [\begin{array}{c|c} \boldsymbol{H}_{\{4,14\}} & \boldsymbol{H}_{\{1,11\}} \end{array} \right]
    \label{eq:top:02}
    \\
    &=\mathrm{rank}\ \left [\begin{array}{c|c|c} \boldsymbol{H}_{\{1,11\}}\boldsymbol{N}_{4,1} & \boldsymbol{H}_{\{2,12\}}\boldsymbol{N}_{4,2} & \boldsymbol{H}_{\{1,11\}} \end{array} \right]
    \label{eq:top:03}
    \\
    &=\mathrm{rank}\ \left [\begin{array}{c|c} \boldsymbol{H}_{\{2,12\}}\boldsymbol{N}_{4,2} & \boldsymbol{H}_{\{1,11\}} \end{array} \right]
    \label{eq:top:04}
    \\
    &=
    \mathrm{rank}\ \boldsymbol{H}_{\{2,12\}}\boldsymbol{N}_{4,2}+ \mathrm{rank} \ \boldsymbol{H}_{\{1,11\}}
    \label{eq:top:05}
    \\
    &=\mathrm{rank}\ \boldsymbol{H}_{\{2,12\}}\boldsymbol{N}_{4,2}+ 2t,
    \label{eq:top:06}
    \end{align}
where \eqref{eq:top:-1} follows from the fact that each $D_i, i\in[6]$ is an independent set. \eqref{eq:top:01} is due to \eqref{eq:pr:thm:21} and \eqref{eq:pr:thm:22}, \eqref{eq:top:02} is because of the invertibility of $\boldsymbol{M}_{21,1}$ and $\boldsymbol{M}_{22,11}$, \eqref{eq:top:03} is due to \eqref{eq:pr:thm2:4-14}, \eqref{eq:top:04} follows from the property of the $\mathrm{rank}$ function by removing $\boldsymbol{H}_{\{1,11\}}\boldsymbol{N}_{4,1}$ which is a subspace of $\boldsymbol{H}_{\{1,11\}}$. \eqref{eq:top:05} is due to the fact that $L_2=\{1,2,11,12\}$ is an independent set, so $\boldsymbol{H}_{\{2,12\}}\boldsymbol{N}_{4,2}$ and $\boldsymbol{H}_{\{1,11\}}$ must be linearly independent. Now, based on \eqref{eq:top:-1}, \dots, \eqref{eq:top:06}, we have
\begin{equation} \label{eq:pr:thm2:N_42}
    \mathrm{rank}\ \boldsymbol{N}_{4,2}=2t.
\end{equation}
Similarly, by taking the same steps in \eqref{eq:top:-1}, $\cdots$, \eqref{eq:pr:thm2:N_42} for the remaining $D_i, i=2,3,4,5,6$, we will have
\begin{align}
     D_2\ \text{and using}\ \eqref{eq:pr:thm:23}, \eqref{eq:pr:thm:24}, &\eqref{eq:pr:thm2:4-14}\rightarrow
    \mathrm{rank}\ \boldsymbol{N}_{4,1}=2t,
    \label{eq:pr:thm2:D_2}
    \\ 
    D_3\ \text{and using}\ \eqref{eq:pr:thm:23}, \eqref{eq:pr:thm:24}, &\eqref{eq:pr:thm2:5-15}\rightarrow 
    \mathrm{rank}\ \boldsymbol{N}_{5,3}=2t, 
    \label{eq:pr:thm2:D_3}
    \\
     D_4\ \text{and using}\ \eqref{eq:pr:thm:25}, \eqref{eq:pr:thm:26}, &\eqref{eq:pr:thm2:5-15}\rightarrow 
    \mathrm{rank}\ \boldsymbol{N}_{5,2}=2t,
    \label{eq:pr:thm2:d4} 
    \\
     D_5\ \text{and using}\ \eqref{eq:pr:thm:21}, \eqref{eq:pr:thm:22}, &\eqref{eq:pr:thm2:6-16}\rightarrow 
    \mathrm{rank}\ \boldsymbol{N}_{6,3}=2t,
    \label{eq:pr:thm2:D_5}
    \\
    D_6\ \text{and using}\ \eqref{eq:pr:thm:25}, \eqref{eq:pr:thm:26}, &\eqref{eq:pr:thm2:6-16}\rightarrow 
    \mathrm{rank}\ \boldsymbol{N}_{6,1}=2t.
    \label{eq:pr:thm2:D_6}
\end{align}
\textbf{\textit{Step 4}}: In this step, we determine the three constraints which must be met on the space of $\boldsymbol{H}_{\{7,17\}}$.
\\
First, let $L_3=\{2,3,12,13\}$. Then,
\begin{align} 
    4t&=\mathrm{rank}\ \boldsymbol{H}_{\{2,3,12,13\}}
    \label{eq:pr:thm2:2-3-12-13}
    \\
    &=\mathrm{rank}\ \boldsymbol{H}_{\{2,12\}} + \mathrm{rank}\ \boldsymbol{H}_{\{3,13\}}
    \label{eq:pr:thm2:2-2:3-13}
    \\
    &=\mathrm{rank}\ \boldsymbol{H}_{\{2,12\}}\boldsymbol{N}_{4,2} + \mathrm{rank}\ \boldsymbol{H}_{\{3,13\}}
    \label{eq:pr:thm2:2-12:N:3-13}
    \\
    &\leq \mathrm{rank}\ \boldsymbol{H}_{\{4,14\}} + \mathrm{rank}\ \boldsymbol{H}_{\{3,13\}}
    \label{eq:pr:thm2:4-14:2-12}
    \\
    &=\mathrm{rank}\ \boldsymbol{H}_{\{3,4,13,14\}}
    \label{eq:pr:thm2:3-4-13-14}
    \\
    &\leq 4t,
    \label{eq:pr:thm2:4t}
\end{align}
where \eqref{eq:pr:thm2:2-3-12-13} and \eqref{eq:pr:thm2:2-2:3-13} follows from the fact that $L_3$ is an independent set. \eqref{eq:pr:thm2:2-12:N:3-13} is due to the invertibility of $\boldsymbol{N}_{4,2}$. \eqref{eq:pr:thm2:4-14:2-12} is because of \eqref{eq:pr:thm2:4-14}, which implies that $\boldsymbol{H}_{\{2,12\}}\boldsymbol{N}_{4,2}$ is a subspace of $\boldsymbol{H}_{\{4,14\}}$, and finally, \eqref{eq:pr:thm2:3-4-13-14} is also due to the \eqref{eq:pr:thm2:4-14}, which indicates that $\boldsymbol{H}_{\{4,14\}}$ is linearly independent of $\boldsymbol{H}_{\{3,13\}}$. Now, based on \eqref{eq:pr:thm2:2-3-12-13}, \dots, \eqref{eq:pr:thm2:4t}, we have
\begin{equation} \label{eq:pr:thm2:3-4-13-14:2}
    \mathrm{rank}\ \boldsymbol{H}_{\{3,4,13,14\}}=4t.
\end{equation}
Moreover, using the same argument in \eqref{eq:pr:thm2:2-3-12-13}, $\cdots$, \eqref{eq:pr:thm2:3-4-13-14:2} for $L_4=\{1,2,11,12\}$ and $L_5=\{1,3,11,13\}$ by considering the facts that (i) each $L_4$ and $L_5$ is an independent set, (ii) $\boldsymbol{N}_{6,1}$ and $\boldsymbol{N}_{5,3}$ are invertible, respectively, due to \eqref{eq:pr:thm2:D_6} and \eqref{eq:pr:thm2:D_3}, (iii) $\boldsymbol{H}_{\{1,11\}}\boldsymbol{N}_{6,1}$ and $\boldsymbol{H}_{\{2,12\}}\boldsymbol{N}_{5,3}$, respectively, are subspace of  $\boldsymbol{H}_{\{6,16\}}$ and $\boldsymbol{H}_{\{5,15\}}$ due to \eqref{eq:pr:thm2:6-16} and \eqref{eq:pr:thm2:5-15}, and finally (iv) $\boldsymbol{H}_{\{6,16\}}$ and $\boldsymbol{H}_{\{5,15\}}$, respectively, are linear independent of $\boldsymbol{H}_{\{2,12\}}$ and $\boldsymbol{H}_{\{1,11\}}$, respectively, due to \eqref{eq:pr:thm2:6-16} and \eqref{eq:pr:thm2:5-15}, we will have
\begin{align}
    \mathrm{rank}\ \boldsymbol{H}_{\{2,6,12,16\}}&=4t,
    \label{eq:pr:thm2:2-6-12-16}
    \\
    \mathrm{rank} \ \boldsymbol{H}_{\{1,5,11,15\}}&=4t.
    \label{eq:pr:thm2:1-5-11-15}
\end{align}
Let $C_{8}^{\prime}=\{1,5,11,15\}\subset C_8=\{1,5,11,15,7,17\}$, $C_{9}^{\prime}=\{2,6,12,16\}\subset C_9=\{2,6,12,16,7,17\}$ and $C_{10}^{\prime}=\{3,4,13,14\}\subset C_{10}=\{3,4,13,14,7,17\}$. Then, for $i=8,9,10$, we have
\begin{align}
    6t&= \mathrm{rank} \ \boldsymbol{H}
    \label{eq"pr:thm2:C_i:first}
    \\
    &\geq \mathrm{rank} \ \boldsymbol{H}_{\{i\}\cup (\{i+10\}\cup C_i)}
    \label{eq:pr:thm2:C_i:8}
    \\
    &= \mathrm{rank} \ \boldsymbol{H}_{\{i+10\}\cup C_i} + t
    \label{eq:pr:thm2:C_i:9}
    \\ 
    &= \mathrm{rank} \ \boldsymbol{H}_{C_i} + 2t
    \label{eq:pr:thm2:C_i:10}
    \\
    &\geq \mathrm{rank} \ \boldsymbol{H}_{C_{i}^{\prime}} + 2t
    \label{eq:pr:thm2:C_i:11}
    \\
    &=6t,
    \label{eq:pr:thm2:C_i:12}
\end{align}
where \eqref{eq"pr:thm2:C_i:first} is because we desire $\lambda_{\mathcal{I}_2,q}=6$, \eqref{eq:pr:thm2:C_i:8} is due to \eqref{eq:rem:dec-con-3}, \eqref{eq:pr:thm2:C_i:9} and \eqref{eq:pr:thm2:C_i:10} are, respectively, because of \eqref{eq:rem:dec-con-1} and the fact that $\{i+10\}\cup C_i\subset B_i$ and $C_i\subset B_{i+10}$. \eqref{eq:pr:thm2:C_i:11} is due to $C_{i}^{\prime}\subset C_i$, and finally, \eqref{eq:pr:thm2:C_i:12} follows from \eqref{eq:pr:thm2:3-4-13-14:2}, \eqref{eq:pr:thm2:2-6-12-16} and \eqref{eq:pr:thm2:1-5-11-15}, respectively for $i=10, 9$ and 8.
Now, based on \eqref{eq"pr:thm2:C_i:first}, \dots, \eqref{eq:pr:thm2:C_i:12}, we have
\begin{equation}
    \mathrm{rank} \ \boldsymbol{H}_{C_i}=\mathrm{rank} \ \boldsymbol{H}_{C_{i}^{\prime}}=4t, \ \ \ \ \forall i=8,9,10,
\end{equation}
which implies that each $\boldsymbol{H}_{C_{i}\backslash C_{i}^{\prime}}$ must be expressed as a linear combination of $\boldsymbol{H}_{C_{i}^{\prime}}$ for $i=8,9,10$. This, respectively, results in
\begin{align}
     \boldsymbol{H}_{\{7,17\}}&=\left [\begin{array}{c|c} \boldsymbol{H}_{\{1,11\}}\boldsymbol{N}_{7,1} & \boldsymbol{H}_{\{5,15\}}\boldsymbol{N}_{7,5} \end{array} \right],
     \label{eq:pr:thm2:7-17-1}
    \\ 
     \boldsymbol{H}_{\{7,17\}}&=\left [\begin{array}{c|c} \boldsymbol{H}_{\{2,12\}}\boldsymbol{N}_{7,2} & \boldsymbol{H}_{\{6,16\}}\boldsymbol{N}_{7,6} \end{array} \right],
    \label{eq:pr:thm2:7-17-2}
    \\
    \boldsymbol{H}_{\{7,17\}}&=\left [\begin{array}{c|c} \boldsymbol{H}_{\{3,13\}}\boldsymbol{N}_{7,3} & \boldsymbol{H}_{\{4,14\}}\boldsymbol{N}_{7,4} \end{array} \right].
    \label{eq:pr:thm2:7-17-3}
\end{align}

\textbf{\textit{Step 5}}: In this step, we illustrate that meeting the three constraints \eqref{eq:pr:thm2:7-17-1}, \eqref{eq:pr:thm2:7-17-2} and \eqref{eq:pr:thm2:7-17-3} on the space of $\boldsymbol{H}_{7,17}$ will lead to a contradiction over any field with characteristic two.
\\
First, in \eqref{eq:pr:thm2:7-17-1}, \eqref{eq:pr:thm2:7-17-2} and \eqref{eq:pr:thm2:7-17-3}, we substitute $\boldsymbol{H}_{4,14}$, $\boldsymbol{H}_{5,15}$, $\boldsymbol{H}_{6,16}$ with their equal term, respectively, in \eqref{eq:pr:thm2:4-14}, \eqref{eq:pr:thm2:5-15} and \eqref{eq:pr:thm2:6-16}. Then, equating the coefficients of each $\boldsymbol{H}_{\{i,i+10\}}, i=1,2,3$, respectively, results in
\begin{align}
  \boldsymbol{N}_{7,1}&=\boldsymbol{N}_{4,1}\boldsymbol{N}_{7,4}=\boldsymbol{N}_{6,1}\boldsymbol{N}_{7,6}, \label{eq:pr:thm2:7-4-6}
  \\
  \boldsymbol{N}_{7,2}&=\boldsymbol{N}_{4,2}\boldsymbol{N}_{7,4}=\boldsymbol{N}_{5,2}\boldsymbol{N}_{7,5}, \label{eq:pr:thm2:7-4-5}
  \\
  \boldsymbol{N}_{7,3}&=\boldsymbol{N}_{5,3}\boldsymbol{N}_{7,5}=\boldsymbol{N}_{6,3}\boldsymbol{N}_{7,6}. \label{eq:pr:thm2:7-5-6}
 \end{align}
Now, because each $\boldsymbol{N}_{4,1}$, $\boldsymbol{N}_{4,2}$, $\boldsymbol{N}_{5,2}$, $\boldsymbol{N}_{5,3}$, $\boldsymbol{N}_{6,1}$ and $\boldsymbol{N}_{6,3}$ is invertible, then, from \eqref{eq:pr:thm2:7-4-6}, \eqref{eq:pr:thm2:7-4-5} and \eqref{eq:pr:thm2:7-5-6}, the column space of all $\boldsymbol{N}_{7,i}$'s for $i\in [7]$ will be equal. Thus, we must have $\mathrm{rank}\ \boldsymbol{N}_{7,i}=2t$, $\forall i\in[7]$, since otherwise $\mathrm{rank}\ \boldsymbol{H}_{\{7,17\}}<2t$, which contradicts \eqref{eq:rem:dec-con-1} for $i=7$ with $B_{7}^{\prime}=\{17\}$. Then, \eqref{eq:pr:thm2:7-4-6}, \eqref{eq:pr:thm2:7-4-5} and \eqref{eq:pr:thm2:7-5-6}, respectively gives
\begin{align}
  \boldsymbol{N}_{4,1}\boldsymbol{N}_{7,4}&=\boldsymbol{N}_{6,1}\boldsymbol{N}_{7,6}, \label{eq:pr:thm2:1:7-4-6}
  \\
  \boldsymbol{N}_{4,2}\boldsymbol{N}_{7,4}&=\boldsymbol{N}_{5,2}\boldsymbol{N}_{7,5}, \label{eq:pr:thm2:2:7-4-5}
  \\
  \boldsymbol{0}_{2t\times 2t}&=\boldsymbol{N}_{6,3}\boldsymbol{N}_{7,6}+\boldsymbol{N}_{5,3}\boldsymbol{N}_{7,5}, \label{eq:pr:thm2:3:7-5-6}
 \end{align}
where \eqref{eq:pr:thm2:3:7-5-6} is achieved from \eqref{eq:pr:thm2:7-5-6} assuming that the field has characteristic two. So,
\begin{equation}
\begin{bmatrix}
   \boldsymbol{N}_{4,1} \\
   \boldsymbol{N}_{4,2} \\
   \boldsymbol{0}_{2t\times 2t}
\end{bmatrix}
\boldsymbol{N}_{7,4}
=
\begin{bmatrix}
   \boldsymbol{0}_{2t\times 2t}                  \\
   \boldsymbol{N}_{5,2}\\
   \boldsymbol{N}_{5,3}
\end{bmatrix}
\boldsymbol{N}_{7,5}
+
\begin{bmatrix}
   \boldsymbol{N}_{6,1} \\
    \boldsymbol{0}_{2t\times 2t}                   \\
   \boldsymbol{N}_{6,3}
\end{bmatrix}
\boldsymbol{N}_{7,6}.
\end{equation}
Thus,
\begin{equation}
    \boldsymbol{H}_{\{4,14\}}\boldsymbol{N}_{7,4}=\boldsymbol{H}_{\{5,15\}}\boldsymbol{N}_{7,5}+\boldsymbol{H}_{\{6,16\}}\boldsymbol{N}_{7,6},
\end{equation}
which means that each $\boldsymbol{H}_{\{4,14\}}$,  $\boldsymbol{H}_{\{5,15\}}$ and $\boldsymbol{H}_{\{6,16\}}$ can be expressed as a linear combination of the other two, resulting in
\begin{equation}\label{eq:pr:thm2:contradiction4}
    \mathrm{rank}\ \boldsymbol{H}_{\{i\}\cup B_i}= \mathrm{rank}\ \boldsymbol{H}_{B_i},  \ \ \ \ \forall i=4,5,6,
\end{equation}
which contradicts the decoding condition in \eqref{eq:dec-cond} for $i=4, 5, 6$. Thus, users $u_4, u_5$ and $u_6$ are not able to decode their requested messages over any field with characteristic two. This completes the proof.
\end{proof}

\begin{prop}
There exists a scalar nonlinear code over the binary field which can achieve the broadcast rate of $\mathcal{I}_2$.
\end{prop}

\begin{proof}
First, note that the set $L=\{1,2,3,11,12,13\}$ is a MAIS set of $\mathcal{I}_2$. So, $\beta_{\mathrm{MAIS(\mathcal{I}_2)}}=6$. Now, we show that $\beta(\mathcal{C}_{\mathcal{I}_2})=6$ for a scalar nonlinear index code $\mathcal{C}_{\mathcal{I}_2}=(\phi_{\mathcal{I}_2}, \{\psi_{\mathcal{I}_2}^i\})$, where the encoder and decoder functions are as follows. First, encoder $\phi_{\mathcal{I}_2}$ maps the messages $x_i, i\in [26]$ to the coded messages $z_j, j\in [6]$ as below
\begin{numcases}{}
   \ \ z_1&$= x_1 \oplus x_4 \oplus  x_6 \oplus  x_7 \oplus  x_{10} \oplus  x_{21}$,\nonumber %\label{eq:1proof:01}
    \\
   \ \  z_2&$= x_{11} \oplus x_{14} \oplus  x_{16} \oplus  x_{17} \oplus  x_{20} \oplus  x_{22} \oplus$ \nonumber
    \\
   \ \   & \ \ \  $(x_{4}x_{6} \oplus x_{4}x_{7} \oplus x_{6}x_{7})$, \nonumber
    %\label{eq:1proof:02} 
    \\
   \ \  z_3&$= x_2 \oplus x_4 \oplus  x_5 \oplus  x_7 \oplus  x_8 \oplus  x_{23}$,\nonumber
    %\label{eq:1proof:03} 
    \\
   \ \  z_4&$= x_{12} \oplus x_{14} \oplus  x_{15} \oplus  x_{17} \oplus  x_{18} \oplus  x_{24} \oplus$ \nonumber
    \\
   \ \  & \ \ \ $(x_{4}x_{5} \oplus x_{4}x_{7} \oplus x_{5}x_{7})$,\nonumber
    %\label{eq:1proof:04} 
    \\
   \ \  z_5&$= x_3 \oplus x_5 \oplus  x_6 \oplus  x_7 \oplus  x_{9} \oplus  x_{25}$,\nonumber
    %\label{eq:1proof:05} 
    \\
   \ \  z_6&$= x_{13} \oplus x_{15} \oplus  x_{16} \oplus  x_{17} \oplus  x_{19} \oplus  x_{26} \oplus$ \nonumber
    \\
   \ \   & \ \ \ $(x_{5}x_{6} \oplus x_{5}x_{7} \oplus x_{6}x_{7})$.\nonumber
    %\label{eq:1proof:06}
\end{numcases}
Now each decoder $\psi_{\mathcal{I}_2}^{i}, i\in[26]$ decodes the requested message $x_i$ using the received coded messages $z_j, j\in[6]$ and the side information $S_i$ as follows:
\begin{itemize}[leftmargin=*]
    \item Users $u_1, u_{11}, u_2, u_{12}, u_3$ and $u_{13}$ can decode their requested message, respectively, from $z_1$, $z_2$, $z_3$, $z_4$, $z_5$ and $z_6$.
    \item Users $u_7$ and $u_{17}$ both can decode $x_7$ from either $z_1$, $z_3$ or $z_5$. Then, $u_{17}$ can decode $x_{17}$ from either $z_2$, $z_4$ or $z_6$.
    \item Users $u_{14}, u_{15}$ and $u_{16}$, respectively, first decode $x_{4}$ from $z_1$, $x_{5}$ from $z_3$ and $x_{6}$ from $z_5$. Then, they can decode, respectively, $x_{14}$ from $z_2$, $x_{15}$ from $z_4$ and $x_{16}$ from $z_6$.
    \item Users $u_{21}$ and $u_{22}$ both first decode both $x_{4}$ from $z_3$ and $x_{6}$ from $z_5$. Then, $u_{21}$ can decode $x_{21}$ from $z_1$. User $u_{22}$ decodes $x_{14}$ from $z_4$ and $x_{16}$ from $z_6$. Now, $u_{22}$ can decode $x_{22}$ from $z_2$.
    \item Users $u_{23}$ and $u_{24}$ both first decode $x_{4}$ from $z_1$ and $x_{5}$ from $z_5$. Then, $u_{23}$ can decode $x_{23}$ from $z_3$. User $u_{24}$ decodes $x_{14}$ from $z_2$ and $x_{15}$ from $z_6$. Now, $u_{24}$ can decode $x_{24}$ from $z_4$.
    \item Users $u_{25}$ and $u_{26}$ both first decode $x_{5}$ from $z_3$ and $x_{6}$ from $z_1$. Then, $u_{25}$ can decode $x_{25}$ from $z_5$. User $u_{26}$ decodes $x_{15}$ from $z_4$ and $x_{16}$ from $z_2$. Now, $u_{26}$ can decode $x_{26}$ from $z_6$.
    \item Users $u_8, u_9$ and $u_{10}$, respectively, decode $(x_5\oplus x_7)$ from $z_5$, $(x_6\oplus x_7)$ from $z_1$ and $(x_4\oplus x_5)$ from $z_3$. Then, they are able to decode their requested messages $x_8$ from $z_3$, $x_9$ from $z_5$ and $x_{10}$ from $z_1$.
    \item User $u_{18}$, first decodes $x_5\oplus x_7$ from $z_5$. Then, it adds $z_4\oplus z_6$ to achieve (after removing the messages in its side information) $x_{18}\oplus (x_4\oplus x_6)(x_5\oplus x_7)$ (note that the term $x_5x_7$ is canceled out). Now, because it has $x_4$ and $x_6$ in its side information and has already decoded $x_5\oplus x_7$, then it will be able to decode its desired message $x_{18}$.
    \item User $u_{19}$, first decodes $x_6\oplus x_7$ from $z_1$. Then, it adds $z_2\oplus z_6$ to achieve (after removing the messages in its side information) $x_{19}\oplus (x_4\oplus x_5)(x_6\oplus x_7)$ (note that the term $x_6x_7$ is canceled out). Now, because it has $x_4$ and $x_5$ in its side information and has already decoded $x_6\oplus x_7$, then it will be able to decode its desired message $x_{19}$.
    \item User $u_{20}$, first decodes $x_4\oplus x_7$ from $z_3$. Then, it adds $z_2\oplus z_4$ to achieve (after removing the messages in its side information) $x_{20}\oplus (x_5\oplus x_6)(x_4\oplus x_7)$ (note that the term $x_4x_7$ is canceled out). Now, because it has $x_5$ and $x_6$ in its side information and has already decoded $x_4\oplus x_7$, then it will be able to decode its desired message $x_{20}$.
    \item Users $u_4, u_5$ and $u_6$ do as follows. Note, in binary field, $x_{i}^{2}=x_i$. User $u_4$ first decodes $(x_4\oplus x_6)$ and $(x_4\oplus x_5)$, respectively, from $z_1$ and $z_3$. Then, it multiplies them to achieve $x_4\oplus x_4x_6\oplus x_4x_5\oplus x_5x_6$. User $u_5$ first decodes $(x_4\oplus x_5)$ and $(x_5\oplus x_6)$, respectively, from $z_3$ and $z_5$. Then, it multiplies them to achieve $x_5\oplus x_4x_6\oplus x_4x_5\oplus x_5x_6$. User $u_6$ first decodes $(x_4\oplus x_6)$ and $(x_5\oplus x_6)$, respectively, from $z_1$ and $z_5$. Then, it multiplies them to achieve $x_6\oplus x_4x_6\oplus x_4x_5\oplus x_5.x_6$.\\
    On the other hand, we add $z_2\oplus z_4\oplus z_6$ to cancel the terms $x_{14}, x_{15}, x_{16}$ and $x_4x_7, x_5x_7, x_6x_7$. Now, it can be observed that users $u_4, u_5$ and $u_6$ will obtain $x_4x_6\oplus x_4x_5\oplus x_5x_6$, and so, they are able to decode their desired message $x_4, x_5$ and $x_6$, respectively.
\end{itemize}

%User $u_4$ first decodes $(x_4\oplus x_6)$ and $(x_4\oplus x_5)$, respectively, from $y_1$ and $y_3$. Then, it multiplies them to achieve $x_4\oplus x_4x_6\oplus x_4x_5\oplus x_5.x_6$ (note that $x_{4}^{2}=x_{4}$). On the other hand, it adds $y_2\oplus y_4\oplus y_6$ to achieve (after removing the messages in its side information) $x_4x_6\oplus x_4x_5\oplus x_5.x_6$ (note that the terms $x_{14}, x_{15}, x_{16}$ and $x_4x_7, x_5x_7, x_6x_7$ are canceled out). So, now it is able to decode its desired message $x_4$.\\
%User $u_5$ first decodes $(x_4\oplus x_5)$ and $(x_5\oplus x_6)$, respectively, from $y_3$ and $y_5$. Then, it multiplies them to achieve $x_5\oplus x_4x_6\oplus x_4x_5\oplus x_5.x_6$ (note that $x_{5}^{2}=x_{5}$). On the other hand, it adds $y_2\oplus y_4\oplus y_6$ to achieve (after removing the messages in its side information) $x_4x_6\oplus x_4x_5\oplus x_5.x_6$ (note that the terms $x_{14}, x_{15}, x_{16}$ and $x_4x_7, x_5x_7, x_6x_7$ are canceled out). So, now it is able to decode its desired message $x_5$.\\
%User $u_6$ first decodes $(x_4\oplus x_6)$ and $(x_5\oplus x_6)$, respectively, from $y_1$ and $y_5$. Then, it multiplies them to achieve $x_6\oplus x_4x_6\oplus x_4x_5\oplus x_5.x_6$ (note that $x_{6}^{2}=x_{6}$). On the other hand, it adds $y_2\oplus y_4\oplus y_6$ to achieve (after removing the messages in its side information) $x_4x_6\oplus x_4x_5\oplus x_5.x_6$ (note that the terms $x_{14}, x_{15}, x_{16}$ and $x_4x_7, x_5x_7, x_6x_7$ are canceled out). So, now it is able to decode its desired message $x_6$.
\end{proof}

\begin{rem}
The specific constraints on the optimal solution of each index coding subinstances $\mathcal{I}_1$ and $\mathcal{I}_2$ were inspired by the constraints on the solution of each network coding subinstances in \cite{Dougherty2005} (which were denoted there by $\mathcal{N}_1$ and $\mathcal{N}_2$). Using similar techniques as in \cite{Dougherty2005}, one can show that linear coding is insufficient over non-commutative rings and modules, where linear operations are well-defined \cite{Connelly20181,Connelly20182}. We leave the technical details for a future version of this work.
\end{rem}

% \begin{rem}
% Our index coding instances $\mathcal{I}_1$ and $\mathcal{I}_2$ were inspired by the same network coding constraints as in \cite{}. Using similar techniques as in \cite{}, one can show that linear coding is insufficient over non-commutative rings and modules, where linear operations are well-defined \cite{Connelly20181,Connelly20182}. We leave the technical details for a future version of this work.
% \end{rem}

\section{Concluding Remarks} \label{sec:06}
In this paper, we addressed the open problem of proving the necessity of nonlinear coding for achieving the symmetric rate of the unicast index coding problem. This proof was made by providing a unicast index coding instance, consisting of two separate subinstances which are connected in a two-way method or no-way (disjoint) method in terms of their side information. We proved that for the first instance linear coding is optimal only over a finite field with characteristic two. However, for the second instance, we proved that linear coding with characteristic two cannot be optimal while an optimal nonlinear code was provided over the binary field. This, in turn, settles the insufficiency of linear codes for unicast setting with symmetric message rate. One main advantage of the structure of our instance is its significant simplicity, having only 36 users while the example in \cite{Effros2015} has 74 messages and 80 users, and also the instance in \cite{Maleki2014} contains more than 200 users.
\\

\IEEEpeerreviewmaketitle

%\vspace{1mm}

\bibliographystyle{IEEEtran}
	\bibliography{References}

% that's all folks
\end{document}